\documentclass[12pt,a4paper]{article}

\usepackage[T1]{fontenc}
\usepackage[utf8]{inputenc}
\usepackage{lmodern}
\usepackage{amsmath,amssymb,amsfonts,amsthm}
\usepackage{enumerate}
\usepackage{indentfirst}
\usepackage[numbers,sort&compress]{natbib}
\usepackage{microtype}
\usepackage{needspace,etoolbox}
\usepackage[a4paper,margin=25mm]{geometry}
\usepackage[hidelinks]{hyperref}
\hypersetup{
  pdftitle={Jacobi last multipliers and metrisability of Lienard differential equations},
  pdfauthor={Maria V. Demina and Varvara G. Nechitailo}}
\DeclareMathOperator{\divr}{div}

\theoremstyle{plain}
\newtheorem{theorem}{Theorem}
\newtheorem{lemma}{Lemma}
\newtheorem{remark}{Remark}

\newtheorem{definition}{Definition}
\BeforeBeginEnvironment{theorem}{\Needspace{5\baselineskip}}
\BeforeBeginEnvironment{lemma}{\Needspace{5\baselineskip}}
\BeforeBeginEnvironment{proposition}{\Needspace{5\baselineskip}}
\BeforeBeginEnvironment{corollary}{\Needspace{5\baselineskip}}
\BeforeBeginEnvironment{definition}{\Needspace{5\baselineskip}}
\BeforeBeginEnvironment{example}{\Needspace{5\baselineskip}}
\BeforeBeginEnvironment{remark}{\Needspace{5\baselineskip}}
\numberwithin{equation}{section}
\numberwithin{lemma}{section}
\numberwithin{theorem}{section}
\numberwithin{proposition}{section}
\numberwithin{corollary}{section}
\numberwithin{definition}{section}
\title{Jacobi last multipliers and metrisability of Li\'{e}nard differential equations}

\author{%
  Maria V. Demina$^{1,2}$\thanks{Corresponding author.
  E-mail: \href{mailto:maria_dem@mail.ru}{\nolinkurl{maria_dem@mail.ru}}.}
  \quad Varvara G. Nechitailo$^1$\\[0.8em]
  \small $^1$HSE University, 34 Tallinskaya Street,\\
  \small 123458 Moscow, Russian Federation\\[0.4em]
  \small $^2$Lomonosov Moscow State University, GSP-1, Leninskie Gory,\\
  \small Moscow, 119991, Russian Federation}
\date{}

\begin{document}
\maketitle

\begin{abstract}
\normalsize
We develop a constructive approach to Jacobi last multipliers and
metrisability of non-autonomous Li\'enard differential equations.
We obtain necessary and sufficient conditions for the existence of
invariants quadratic in the first derivative whose cofactors are
proportional to the coefficient of the first derivative. These invariants
yield Jacobi last multipliers and, in particular cases, first integrals
that are ratios of quadratic polynomials in the first derivative.
For a distinguished value of the proportionality constant,
non-degenerate invariants characterize metrisability: the solution curves
are unparameterized geodesics of a metric on a two-dimensional surface.
For prescribed equation coefficients, the determining conditions reduce
to linear ordinary differential equations for three functions of one
variable, possibly supplemented by algebraic constraints.
We give a constructive description of the metrics associated with
autonomous Li\'enard equations when their geodesic flows have no first
integrals linear in momenta. Applications include equations satisfying
the Chiellini integrability condition, which yield metrics with rational
or transcendental first integrals in momenta. We also construct a
positive-definite metric with a first integral that is non-Liouvillian
in the momentum ratio, and establish local metrisability of the
Ermakov--Pinney equation.
\end{abstract}


\section{Introduction}\label{S:Int}

Non-autonomous Li\'{e}nard differential equations
\begin{equation}
\label{Eq:Lienard_Int}
y_{xx}=f(x,y)y_x+g(x,y)
\end{equation}
 arise in a variety of applications and have been intensively studied in recent years. In particular, equations \eqref{Eq:Lienard_Int} model oscillatory processes. In this case the functions $f(x,y)$  and $g(x,y)$ describe the damping coefficient and  the restoring force, respectively. Much is known about the  invariants and first integrals of autonomous Li\'{e}nard equations, see, for example \cite{Llibre06, harko2014class, ruiz2018integrability, demina2023integrability, demina2025integrability, la2026c, demina2026non}. Analogous problems for non-autonomous Li\'{e}nard equations are considerably more complicated. Conditions for the existence of first integrals linear, fractional-linear, and quadratic in $z=y_x$ are presented in \cite{muriel2009second, muriel2011second, sinelshchikov2019lax}, respectively. The aim of the present paper is to classify equations \eqref{Eq:Lienard_Int} that have  Jacobi last  multipliers of a special form. The existence of such Jacobi last multipliers is closely related to the problem of determining geodesics on two-dimensional surfaces.


Let $z=y_x$. We  associate the following vector field with  equation \eqref{Eq:Lienard_Int}:
\begin{equation}
 \label{Eq:Lienard_vector_field}
 \mathcal{X}=\partial_x+z\partial_y+(f(x,y)z+g(x,y))\partial_z.
\end{equation}
 A continuously differentiable  function $H(x,y,z)$: $D\subseteq\mathbb{C}^3\rightarrow \mathbb{C}$ is called an invariant of equation \eqref{Eq:Lienard} if it satisfies the linear partial differential equation  $\mathcal{X}H=\lambda(x,y,z)H$, where the function $\lambda(x,y,z)$ is commonly referred to as the cofactor of the invariant $H(x,y,z)$. According to this definition, any continuously differentiable  function is an invariant. Of course, only invariants with certain properties receive special attention,  for example, invariants with cofactors polynomial with respect to $z$.

 A first integral $I(x,y,z)$ of equation \eqref{Eq:Lienard_Int} is defined as  an invariant with the zero cofactor.  Equation \eqref{Eq:Lienard_Int} is  integrable if we know two functionally independent  first integrals in $D\subseteq\mathbb{C}^3$.   An invariant   $M(x,y,z)$  with the cofactor $\lambda=-\divr \mathcal{X}$, where $\divr \mathcal{X}=f$, is called a Jacobi last multiplier. Jacobi last multipliers are of considerable interest in the qualitative and analytic theory of ordinary differential equations.

A Jacobi last multiplier $M$ makes the rescaled vector field
$M\mathcal{X}$ divergence-free. For a second-order ODE,
a Jacobi last multiplier can be used to find a second first
integral by quadratures whenever one first integral is known.
Let $I(x,y,z)$ be a first integral and $M(x,y,z)$ be a Jacobi
last multiplier of equation~\eqref{Eq:Lienard_Int}.
We work locally where $M\neq0$ and $I_z\neq0$.
Let $z=\psi(x,y,C)$ be the solution of $I(x,y,z)=C$
given by the implicit function theorem, and set $R(x,y,C)=(M/I_{z})|{z=\psi(x,y,C)}$.
For each fixed $C$, the differential form $R(x,y,C)\left(dy-\psi(x,y,C)\,dx\right)$
is closed and hence locally exact. Let $K(x,y,C)$ be
a local primitive, with $C$ treated as a parameter.
Then $J(x,y,z)=K(x,y,I(x,y,z))$ is a first integral of equation~\eqref{Eq:Lienard_Int}
functionally independent of $I$. In addition, we note that Jacobi last  multipliers can be used to develop a Lagrangian formalism for the related equation~\cite{nucci2008jacobi}. A comprehensive review of the other geometric properties and applications of the Jacobi last multipliers is given in \cite{carinena2021jacobi}.


This article is organized as follows. In Section \ref{S:Invariants} we present a general method of finding quadratic in $z=y_x$ invariants with cofactors proportional to $f$  and consider some properties of equations \eqref{Eq:Lienard_Int} possessing such invariants. Section \ref{S:Metrisability} is devoted to the metrisability problem. In Section \ref{S:Autonomous} we consider in detail the autonomous  Li\'{e}nard equations and find the necessary and sufficient conditions for the existence of the non-autonomous quadratic invariants with the prescribed cofactors. We present the solution of the metrisability problem for autonomous  Li\'{e}nard equations whenever the related metrics do not have linear in momenta first integrals. In addition, in Section \ref{S:Autonomous} we find integrable metrics resulting from the autonomous  Li\'{e}nard equations satisfying the Chiellini integrability condition \cite{Chiellini01} and another condition from \cite{demina2026non}. In Section \ref{S:f_zero} we study Li\'{e}nard equations with $f\equiv0$ and prove that the  Ermakov--Pinney equation is metrisable.

\section{Quadratic invariants and Jacobi last multipliers}\label{S:Invariants}



Let $U_x$ be a neighborhood of zero in the $x$-plane and $S_y$ be a punctured sector in the $y$-plane with  an angle less
than $2\pi$. Fix the angular interval of $S_y$ and a branch of $\log y$
on this sector.
We work with sectorial germs represented by functions holomorphic on
product domains $U_x\times S_y$. Two representatives define the same
germ if they agree after shrinking $U_x$ and the radius of $S_y$,
while keeping the angular interval and the branch of the logarithm
fixed. Let $\mathbb C\{x,t\}$ denote the ring of convergent power
series at $(x,t)=(0,0)$.

\begin{definition}
Let $\mathcal A$ be the class of sectorial germs admitting a
representation
\begin{equation}\label{Eq:genLS_gen}
h(x,y)=
y^{p/N}\sum_{j=0}^{M}
a_j\!\left(x,y^{1/N}\right)(\log y)^j,
\quad
a_j\in\mathbb C\{x,t\},
 N\in\mathbb{N},
 p\in\mathbb{Z},
 M\in\mathbb{N}\cup\{0\}.
\end{equation}
The integers $N$, $p$, and $M$ may depend on the germ.
\end{definition}

The class $\mathcal A$ is a ring equipped with the commuting
derivations $\partial_x$ and $\partial_y$, which act term by term.
It is also closed under taking primitives with respect to $y$.
Let $\mathcal B=\mathbb C\{x\}$ be the ring of holomorphic germs at
$x=0$. For each $n\in\mathbb{N}\cup\{0\}$, define the
$\mathcal B$-submodule
\begin{equation}
\mathcal P_n=
\left\{
\sum_{j=0}^{n}a_j(x)y^j:
a_j\in\mathcal B
\right\}
\subset\mathcal A.
\end{equation}
Two elements $h_1,h_2\in\mathcal A$ are equivalent modulo
$\mathcal P_n$ if
\begin{equation}
h_1\sim_n h_2
\quad\Longleftrightarrow\quad
h_1-h_2\in\mathcal P_n.
\end{equation}
We denote the resulting quotient $\mathcal B$-module by
$\mathcal A_n=\mathcal A/\mathcal P_n$.

In particular, representatives of the same classes in $\mathcal A_1$
and $\mathcal A_0$, respectively, are related by
\begin{equation}
\widetilde F=F+a(x)y+b(x),
\qquad
\widetilde G=G+c(x),
\qquad a,b,c\in\mathcal B.
\end{equation}
We use the same symbols $F$ and $G$ for these classes and for chosen
representatives in $\mathcal A$. The derivatives $F_{yy}$ and $G_y$
are independent of this choice.

Our aim in this section is to obtain the necessary and sufficient
conditions for Li\'enard differential equations
\eqref{Eq:Lienard_Int} to admit nonzero invariants that are polynomial
of degree at most two in $z=y_x$:
\begin{equation}
\label{Eq:genLS}
H(x,y,z)=C_0(x,y)z^2+2C_1(x,y)z+C_2(x,y),
\qquad C_0,C_1,C_2\in\mathcal A.
\end{equation}
For convenience, we write these equations as
\begin{equation}
\label{Eq:Lienard}
y_{xx}=F_{yy}y_x+G_y.
\end{equation}
We seek invariants with the cofactor $\lambda=-\alpha F_{yy}$, where
$\alpha\in\mathbb C$ is a constant. For $\alpha=0$, such an invariant
is a first integral. For $\alpha\ne0$, it yields a Jacobi last
multiplier $M=H^{1/\alpha}$
locally on any open set where $H\ne0$, after choosing a branch of
$\log H$; see Lemma~\ref{L:LYM} below.

The following theorem is the main result of this section.

\begin{theorem}\label{T:Inv}
Let $\alpha\in\mathbb C$, $F\in\mathcal A_1$,  $G\in\mathcal A_0$,
and choose representatives of $F$ and $G$ in $\mathcal A$.
The Li\'enard differential equation~\eqref{Eq:Lienard} admits a
nonzero invariant of the form~\eqref{Eq:genLS} with cofactor
$\lambda=-\alpha F_{yy}$ if and only if there exist
$l,s,r\in\mathcal B$, not all identically zero, such that
\begin{equation}
\label{Eq:metrisability_Cond1}
(A)_{\alpha}:\quad
\left(\partial_x+\alpha F_{yy}\right)\Psi
+\left(2s-(\alpha+2)lF_y-l_xy\right)G_y=0,
\end{equation}
where
\begin{equation}
\label{Eq:metrisability_Psi}
\begin{aligned}
\Psi(x,y)={}&
\frac{(\alpha+1)(\alpha+2)l}{2}(F_y)^2
+(\alpha+1)(l_xy-2s)F_y
+(\alpha+2)lF_x\\
&+l_xF-2lG+\frac{l_{xx}}{2}y^2-2s_xy+r.
\end{aligned}
\end{equation}
The coefficients of the corresponding invariant are
\begin{equation}
\label{Eq:metrisability_Cond1_fg}
C_0=l,\qquad
C_1=s-\frac{l_x}{2}y-\frac{(\alpha+2)l}{2}F_y,
\qquad
C_2=\Psi.
\end{equation}
Every invariant of the form~\eqref{Eq:genLS} with this cofactor is
obtained in this way.
\end{theorem}

\begin{proof}
The condition that $H$ have cofactor $\lambda=-\alpha F_{yy}$ is
\begin{equation}
\label{Eq:Main_PDE}
H_x+zH_y+(F_{yy}z+G_y)H_z+\alpha F_{yy}H=0.
\end{equation}
Substituting~\eqref{Eq:genLS} into this equation and equating
coefficients of powers of $z$ gives
\begin{equation}
\label{Eq:Main_system}
\begin{gathered}
(C_0)_y=0,\quad
2(C_1)_y+(C_0)_x+(\alpha+2)F_{yy}C_0=0,\\
(C_2)_y+2(C_1)_x+2(\alpha+1)F_{yy}C_1+2G_yC_0=0,\\
(C_2)_x+\alpha F_{yy}C_2+2G_yC_1=0.
\end{gathered}
\end{equation}
The first equation gives $C_0=l(x)$ with $l\in\mathcal B$.
Integrating the second equation with respect to $y$ yields
\begin{equation}
C_1=s-\frac{l_x}{2}y-\frac{(\alpha+2)l}{2}F_y,
\qquad s\in\mathcal B.
\end{equation}
Substituting these expressions into the third equation and
integrating with respect to $y$ gives $C_2=\Psi$, where $\Psi$ is
defined by~\eqref{Eq:metrisability_Psi} and $r\in\mathcal B$ is the
integration constant. Substituting $C_1$ and $C_2$ into the last
equation gives precisely $(A)_\alpha$.

Conversely, suppose that $l,s,r\in\mathcal B$ satisfy $(A)_\alpha$.
Then the coefficients in~\eqref{Eq:metrisability_Cond1_fg} belong to
$\mathcal A$ and satisfy all four equations in
\eqref{Eq:Main_system}. Hence the corresponding $H$ satisfies
\eqref{Eq:Main_PDE}. Moreover, $H\equiv0$ if and only if
$l=s=r\equiv0$, so the stated nonvanishing condition ensures that
the invariant is nonzero.
\end{proof}

\begin{remark}
The criterion is independent of the chosen representatives of $F$
and $G$. Indeed, setting $\widetilde F=F+ay+b$, $\widetilde G=G+c$, $a,b,c\in\mathcal B$,
we get
\begin{equation}
\begin{aligned}
\widetilde l=l,\,
\widetilde s=s+\frac{\alpha+2}{2}la,\,
\widetilde r=r+2(\alpha+1)as
+\frac{(\alpha+1)(\alpha+2)}{2}la^2
-(\alpha+2)lb_x-l_xb+2lc.
\end{aligned}
\end{equation}
A direct substitution gives
$\widetilde C_0=C_0$, $\widetilde C_1=C_1$, and
$\widetilde\Psi=\Psi$. Thus, the invariant and condition
$(A)_\alpha$ are unchanged when the parameters are transformed as
above.
\end{remark}

\begin{remark}
The invariant has degree exactly two in $z$ if and only if
$l\not\equiv0$. If $l\equiv0$ and $s\not\equiv0$, it has degree
exactly one. If $l=s\equiv0$, then $H=r(x)$ and $(A)_\alpha$
reduces to $r_x+\alpha F_{yy}r=0$.
In particular, when $\alpha=0$, this last case gives only constant
first integrals. Thus, for $\alpha=0$, the criterion for a
nonconstant first integral of degree at most two is obtained by
requiring that $l$ and $s$ are not both identically zero. For an
invariant of degree exactly two, for any $\alpha$, the required
condition is $l\not\equiv0$.
\end{remark}

Now we will describe a  method of working with condition $(A)_{\alpha}$. For prescribed $f$, $g\in \mathcal{A}$ in \eqref{Eq:Lienard_Int}, and
$\alpha\in\mathbb{C}$, the unknown functions $l,s,r\in\mathcal B$ depend
only on $x$. Consequently, condition $(A)_\alpha$ reduces
to a system of ordinary differential equations, possibly
including algebraic constraints.

\textit{Step 1.} Expand $f$ and $g$ in the form~\eqref{Eq:genLS_gen}. Choose representatives $F\in\mathcal{A}_1$ and $G\in\mathcal{A}_0$ via the relations $F=(\partial_y^{-1})^2f$ and $G=(\partial_y^{-1})g$.
Then $F_{yy}=f$ and $G_y=g$.

\textit{Step 2.}
Substitute these representatives into $(A)_\alpha$ and denote
its left-hand side by
\begin{equation}
R_\alpha[l,s,r]
=
(\partial_x+\alpha F_{yy})\Psi
+
\bigl(2s-(\alpha+2)lF_y-l_xy\bigr)G_y.
\end{equation}
After passing to a common denominator of the fractional
exponents and collecting like terms, write
\begin{equation}
R_\alpha[l,s,r](x,y)
=
\sum_{k=0}^{+\infty}\sum_{m=0}^{M}
E_{k,m}[l,s,r](x)
y^{(k-p)/N}(\log y)^m.
\end{equation}
Uniqueness of the expansion implies that $(A)_\alpha$ is
equivalent to the system
\begin{equation}
(S)_\alpha:\qquad
E_{k,m}[l,s,r](x)\equiv0,
\qquad k\geq0,\quad 0\leq m\leq M.
\end{equation}
For fixed $f$, $g$, and $\alpha$, every equation is linear
and homogeneous in $l$, $s$, $r$, and their derivatives.
The equations may be examined in increasing order of the
exponent $(k-p)/N$, taking all logarithmic coefficients
at each exponent into account.

\textit{Step 3.}
Solve a finite subsystem of $(S)_\alpha$, retaining its
general solution and all relevant cases. Record any
compatibility conditions on the prescribed coefficients
and on $\alpha$. The coefficient equations corresponding to $y^2(\log y)^0$, $y(\log y)^0$, and $y^0(\log y)^0$
provide a normal subsystem, as shown below, and hence form
a convenient starting point.

When constructing nonzero invariants, retain only triples
satisfying $(l,s,r)\not\equiv(0,0,0)$.
For invariants of degree exactly two in $z$, impose the
additional condition $l\not\equiv0$.

\textit{Step 4.}
Substitute the resulting candidate family into the full
condition $(A)_\alpha$. Verify the identity either directly
or through its convergent expansion. If it does not hold
for the entire candidate family, impose the additional
compatibility conditions or return to Step 3 and include
further equations of $(S)_\alpha$.

\textit{Step 5.}
Classify the resulting solutions using the degeneracy
condition $(B)_\alpha$ defined below. For invariants of
degree exactly two in $z$, this condition characterizes
a double root on the domain where $l(x)\ne0$.

If both expansions in Step 1 terminate, then $(S)_\alpha$
is finite. In this case, all its equations may be considered
simultaneously, and it may be convenient to begin with those
corresponding to the largest exponents.

Let us find a uniform bound on the dimension of the
solution space. Inspection of $(A)_\alpha$ gives
\begin{equation}
R_\alpha[l,s,r]
=
\frac{y^2}{2}l_{xxx}
-2ys_{xx}
+r_x
+
 L_\alpha
\bigl(l,l_x,l_{xx},s,s_x,r\bigr),
\end{equation}
where $ L_\alpha$ is a linear expression in its six
arguments, with coefficients in $\mathcal A$.
Extracting the coefficients of
$y^2(\log y)^0$, $y(\log y)^0$, and $y^0(\log y)^0$ yields
\begin{equation}
\begin{pmatrix}
\frac12&0&0\\
0&-2&0\\
0&0&1
\end{pmatrix}
\begin{pmatrix}
l_{xxx}\\
s_{xx}\\
r_x
\end{pmatrix}
+
K_\alpha(x)
\begin{pmatrix}
l\\
l_x\\
l_{xx}\\
s\\
s_x\\
r
\end{pmatrix}
=0,
\end{equation}
where the entries of $K_\alpha$ belong to $\mathcal B$.
Since the first matrix is invertible, these equations are
equivalent to a normal homogeneous system
\begin{equation}
Z_x=M_\alpha(x)Z,
\qquad
Z=(l,l_x,l_{xx},s,s_x,r)^{\mathsf T},
\end{equation}
with holomorphic coefficients.

Every solution of $(S)_\alpha$ satisfies this normal system
and is therefore uniquely determined by $Z(0)$. Thus, the
space of solutions, including the zero solution,
satisfies
\begin{equation}
\dim_{\mathbb C}\operatorname{Sol}(S_\alpha)\leq6.
\end{equation}
The full system is in general overdetermined: its remaining
equations impose linear conditions on the six initial
constants.

We define the index of $(S)_\alpha$ and of the corresponding
Li\'enard differential equation \eqref{Eq:Lienard_Int}   by
\begin{equation}
d_\alpha=d_\alpha(f,g)
=
\dim_{\mathbb C}\operatorname{Sol}(S_\alpha),
\qquad
d_\alpha\in\{0,1,\ldots,6\}.
\end{equation}
By Theorem~\ref{T:Inv}, the condition $d_\alpha=0$ is equivalent
to the absence of nonzero relative invariants of degree at
most two in $z$ with cofactor $\lambda=-\alpha f$. For $\alpha=0$, the triple $(l,s,r)=(0,0,1)$ always gives the
constant invariant $H=1$, so $d_0\geq1$. In particular, when
$f\equiv0$, the cofactor is zero for every $\alpha$, and
$d_\alpha=d_0$.


We say that a solution of system $(S)_{\alpha}$ is degenerate if the
coefficients of the corresponding polynomial~\eqref{Eq:genLS}
satisfy $C_0C_2-C_1^2\equiv0$, i.e., if its quadratic discriminant
$4(C_1^2-C_0C_2)$ vanishes identically. This definition also applies
when the degree of the polynomial is less than two. The corresponding
invariants will be called degenerate. The degeneracy condition reads as
\begin{equation}
\label{Eq:metrisability_Cond_degenerate}
\begin{gathered}
(B)_{\alpha}:\quad
l\Psi-\left(s-\frac{l_x}{2}y-\frac{(\alpha+2)l}{2}F_y\right)^2=0.
\end{gathered}
\end{equation}
Suppose that $l(x)\not\equiv0$. Then, on the set where $l(x)\ne0$,
expressing the function $G$ from
\eqref{Eq:metrisability_Cond_degenerate} yields
\begin{equation}
\label{Eq:metrisability_Cond2}
\begin{gathered}
G=\frac{(\alpha+2)\alpha}{8}(F_y)^2
+\frac{\alpha(l_xy-2s)}{4l}F_y
+\frac{\alpha+2}{2}F_x+\frac{l_x}{2l}F\\
+\left(\frac{l_{xx}}{4l}-\frac{(l_x)^2}{8l^2}\right)y^2
+\left(\frac{l_xs}{2l^2}-\frac{s_x}{l}\right)y
+\frac{r}{2l}-\frac{s^2}{2l^2}.
\end{gathered}
\end{equation}
 When this formula is used to define
$G\in\mathcal A$, the right-hand side must extend to an element of
$\mathcal A$; in particular, all apparent poles in $x$ must be
removable.

For $l\not\equiv0$, condition $(B)_{\alpha}$ also implies
$(A)_{\alpha}$. Indeed, denoting the left-hand sides of these
conditions by $B_{\alpha}$ and $A_{\alpha}$, respectively, we obtain
the identity $l^2A_{\alpha}=
l\,\partial_x B_{\alpha}
+(\alpha lF_{yy}-l_x)B_{\alpha}
-C_1\,\partial_y B_{\alpha}$,
where $C_1$ is given by~\eqref{Eq:metrisability_Cond1_fg}.
Thus, $B_{\alpha}\equiv0$ implies $A_{\alpha}\equiv0$.
Consequently, if condition~\eqref{Eq:metrisability_Cond2} holds with
$G\in\mathcal A$, the corresponding invariant~\eqref{Eq:genLS}
has a double root with respect to $z$ wherever $l(x)\ne0$ and reads as
\begin{equation}
\label{Eq:metrisability_Invariant_degenerate}
H(x,y,z)=
\frac{\left(2lz-(\alpha+2)lF_y-l_xy+2s\right)^2}{4l}.
\end{equation}

Now let us study the properties of Li\'{e}nard differential
equations~\eqref{Eq:Lienard} satisfying the conditions of
Theorem~\ref{T:Inv}.

\begin{lemma}\label{L:LYM}
The following statements are valid.
\begin{enumerate}

\item
If $H=C_0z^2+2C_1z+C_2$ is a nonzero invariant of
equation~\eqref{Eq:Lienard} with the cofactor
$\lambda=-\alpha F_{yy}$, where
$\alpha\in\mathbb C\setminus\{0\}$, then
$M=H^{1/\alpha}$ is a Jacobi last multiplier of the equation under
consideration, defined locally where $H\ne0$ after choosing a branch
of $\log H$.

\item
Let $\alpha\in\mathbb C$ be some fixed number. If
$H_j=C_0^{(j)}z^2+2C_1^{(j)}z+C_2^{(j)}$, $j=1,2$, are
invariants of equation~\eqref{Eq:Lienard} with the cofactor
$\lambda=-\alpha F_{yy}$ and are linearly independent over
$\mathbb C$, then $I=H_1/H_2$ is a nonconstant first
integral of the equation under consideration on the set where
$H_2\ne0$.

\item
Let $F_{yy}=\mu(x)$, where $\mu\in\mathcal B$.
In this case, for every $\alpha\in\mathbb C$, system $(S)_{\alpha}$
has a degenerate solution that gives the invariant $H_1=\exp(-\alpha\partial_x^{-1}\mu)$
of equation~\eqref{Eq:Lienard} with the cofactor
$\lambda=-\alpha\mu(x)$. Here $\partial_x^{-1}\mu$ denotes a fixed
holomorphic primitive of $\mu$. If, in addition, there exists
another invariant $H_2=C_0z^2+2C_1z+C_2$ with the same cofactor,
linearly independent of $H_1$ over $\mathbb C$, then $I=\exp\left(\alpha\partial_x^{-1}\mu\right)
\left(C_0z^2+2C_1z+C_2\right)$
is a nonconstant first integral of the equation under consideration.

\end{enumerate}
\end{lemma}

\begin{proof}
 Statement $1$ can be verified by direct computations. Let us prove statement~$2$. Since
$\mathcal{X}H_j=-\alpha F_{yy}H_j$, $j=1,2$, we find that $\mathcal{X}(H_1/H_2)=0$. The linear independence of $H_1$ and $H_2$ over $\mathbb C$
ensures that this ratio is nonconstant.
Thus, $I=H_1/H_2$ is a first integral of the related equation.

 Now let us prove statement~$3$. Substituting $F=\mu y^2/2$ into condition $(A)_{\alpha}$, we see that this condition is always satisfied if we set $l\equiv0$, $s\equiv0$, and $r=\exp(-\alpha \partial_x^{-1}\mu)$. Analyzing condition $(B)_{\alpha}$, we conclude that this solution is degenerate. The related invariant takes the form $H_1=\exp(-\alpha \partial_x^{-1}\mu)$. Assuming that there exists another linearly independent invariant $H_2=C_0z^2+2C_1z+C_2$, we use statement $2$ to prove that the ratio $H_2/H_1$ is a first integral of the related equation.
 \end{proof}

 If $\alpha=0$, then the invariants of Theorem \ref{T:Inv} are first integrals quadratic with respect to the first derivative. In the special case $C_0=1$ the problem of finding such first integrals was considered in \cite{sinelshchikov2019lax}.  While in assumptions of statement $2$ from Lemma \ref{L:LYM} the first integral $I(x,y,z)=H_1/H_2$ is  fractional-quadratic with respect to the first derivative provided that $C_0^{(1)}C_0^{(2)}\not\equiv0$.

\textbf{Example 1.} Consider the Li\'{e}nard differential equation
\begin{equation}
\label{Eq:Lenard_Ex1_1}
y_{xx}=(4y^2+a(x))y_x-y^5-a(x)y^3+b(x)y,
\end{equation}
where $a,b\in\mathcal B$. Setting $\alpha=-1/2$ and substituting
\begin{equation}
\label{Eq:Lenard_Ex1_2}
F=\frac{1}{3}y^4+\frac{a(x)}{2}y^2,\qquad
G=-\frac{1}{6}y^6-\frac{a(x)}{4}y^4+\frac{b(x)}{2}y^2
\end{equation}
into condition $(A)_{-\frac12}$, we see that system
$(S)_{-\frac12}$ is satisfied if and only if
\begin{equation}
\label{Eq:Lenard_Ex1_3}
s=0,\qquad
r=c_4\exp[(\partial_x^{-1}a)/2],
\qquad c_4\in\mathbb C,
\end{equation}
and $l$ solves the equation
\begin{equation}
\label{Eq:Lenard_Ex1_4}
\begin{gathered}
l_{xxx}+\frac{3a}{2}l_{xx}
+\left(\frac{7a_x}{2}-\frac{a^2}{4}-4b\right)l_x\\
+\left(\frac{3a_{xx}}{2}+\frac{3aa_x}{4}
-2b_x-\frac{3a^3}{8}-2ab\right)l
-4c_4\exp[(\partial_x^{-1}a)/2]=0.
\end{gathered}
\end{equation}
Here $\partial_x^{-1}a$ denotes a fixed holomorphic primitive of
$a$. The general solution of this equation can be represented as
\begin{equation}
\label{Eq:Lenard_Ex1_5}
l(x)=\sum_{j=1}^3c_jl_j(x)+c_4w(x),
\qquad c_1,c_2,c_3\in\mathbb C.
\end{equation}
Here $l_1$, $l_2$, and $l_3$ are linearly independent solutions of
the homogeneous equation obtained from
\eqref{Eq:Lenard_Ex1_4} by setting $c_4=0$. The function $w(x)$
is a fixed particular solution of~\eqref{Eq:Lenard_Ex1_4} with
$c_4=1$, so that $c_4w(x)$ is a particular solution for an arbitrary
value of $c_4$.

We conclude that the index $d_{-\frac12}$ of system
$(S)_{-\frac12}$ equals $4$. The space of invariants of the
form~\eqref{Eq:genLS} with the cofactor
$\lambda=(4y^2+a)/2$ is four-dimensional, and every nonzero
invariant in this space has degree two with respect to $z$.
These invariants take the form
\begin{equation}
\label{Eq:Lenard_Ex1_6}
\begin{gathered}
H=lz^2-
\left[\left(2y^2+\frac{3a}{2}\right)l+l_x\right]yz
+ly^6+\left(l_x+\frac32al\right)y^4\\
+\left[\frac{l_{xx}}{2}+al_x+
\left(\frac{3a_x}{4}+\frac{3a^2}{8}-b\right)l\right]y^2
+c_4\exp[(\partial_x^{-1}a)/2].
\end{gathered}
\end{equation}
Let $H_1,\ldots,H_4$ be linearly independent invariants over
$\mathbb C$. Since they have the same cofactor, the functions $I_1=H_1/H_2$ and $I_2=H_3/H_4$
are fractional-quadratic first integrals of
equation~\eqref{Eq:Lenard_Ex1_1}.
These first integrals are functionally independent in general
position.

Let $B_1,B_2\in\mathbb C$ and denote by $R$ the resultant of
$H_1-B_1H_2$ and $H_3-B_2H_4$ with respect to $z$.
 Eliminating $z$ gives
$R=0$. Thus, the corresponding local branches of $R=0$ implicitly
define the general solution.

 \section{Metrisability problem}\label{S:Metrisability}
It is known that solutions of a second-order ordinary differential equation  cubic with respect to the first derivative
\begin{equation}
\label{Eq:ODE_projective}
y_{xx} = A_3(x,y)(y_x)^3+A_2(x,y)(y_x)^2+A_1(x,y)y_x+A_0(x,y)
\end{equation}
define geodesics of a Levi--Civita connection  of a  metric  on some two dimensional surface provided that certain restrictions on the coefficient-functions $A_0(x,y)$, $\dots$, $A_3(x,y)$ are satisfied. In this case the related equation is called metrisable. The necessary and sufficient  conditions of metrisability were derived by R. Liouville \cite{liouville1889invariants}. Equation~\eqref{Eq:ODE_projective} is metrisable if and only if the linear system
 with respect to the functions $\psi_1(x,y)$, $\psi_2(x,y)$, $\psi_3(x,y)$
\begin{equation}
 \label{Eq:Geodesic_main_PDE}
 \begin{gathered}
(\psi_{1})_x=\frac23 A_1\psi_1-2A_0\psi_2,\quad
(\psi_3)_y=2 A_3\psi_2-\frac23A_2\psi_3,\hfill\\
(\psi_{1})_y+2(\psi_{2})_x=\frac43 A_2\psi_1-\frac23 A_1\psi_2-2A_0\psi_3,\hfill\\
(\psi_{3})_x+2(\psi_{2})_y=2 A_3\psi_1-\frac43 A_1\psi_3+\frac23A_2\psi_2\hfill
\end{gathered}
\end{equation}
has a non-degenerate solution in a neighborhood of some point. The non-degeneracy condition takes the form $\Delta=\psi_1\psi_3-\psi_2^2\neq0$. The functions $\psi_1(x,y)$, $\ldots$, $\psi_3(x,y)$ and the coefficients of the metric $ds^2=g_{jk}dx^jdx^k$, $(x^1,x^2)=(x,y)$ are related by the equalities $g_{11}=\psi_1/\Delta^2$, $g_{12}=\psi_2/\Delta^2$, and $g_{22}=\psi_3/\Delta^2$.

We note that the geodesic system reads as
\begin{equation}
\label{Eq:Geodesic_system}
\ddot{x}^m+\Gamma^m_{jk}\dot{x}^j\dot{x}^k=0,\quad m=1,2,\quad (x^1,x^2)=(x,y),
\end{equation}
where the Christoffel symbols satisfy the relations
\begin{equation}
\label{Eq:Christoffel}
\Gamma^m_{jk}=\frac12g^{ml}\left(\frac{\partial g_{jl}}{\partial x^k}+\frac{\partial g_{kl}}{\partial x^j}-\frac{\partial g_{jk}}{\partial x^l}\right),\quad g^{mk}g_{kj}=\delta_j^m,\quad m,j,k=1,2.
\end{equation}
Here the dot denotes the derivative with respect to $t$ and $\delta_j^m$ is the Kronecker delta.  Projecting  system \eqref{Eq:Geodesic_system} onto the plane $(x,y)$ yields equation \eqref{Eq:ODE_projective}, where $A_0=-\Gamma^2_{11}$, $A_1=\Gamma^1_{11}-2\Gamma^2_{12}$, $A_2=2\Gamma^1_{12}-\Gamma^2_{22}$, and $A_3=\Gamma^1_{22}$. System \eqref{Eq:Geodesic_system} is Hamiltonian
\begin{equation}
\label{Eq:Ham_system}
\dot{x}^j=\frac{\partial H}{\partial p_j},\quad \dot{p}_j=-\frac{\partial H}{\partial x^j},\quad H=\frac{g^{jk}p_jp_k}{2}
\end{equation}
with the standard Poisson bracket
\begin{equation}
\label{Eq:Poisson_bracket}
\{F,H\}=\sum_{j=1}^2\left(\frac{\partial F}{\partial x^j}\frac{\partial H}{\partial p_j}-\frac{\partial F}{\partial p_j}\frac{\partial H}{\partial x^j}\right).
\end{equation}
A function $J(x^1,x^2,p_1,p_2)$ is a first integral of
system~\eqref{Eq:Ham_system} if and only if it satisfies
$\{J,H\}=0$, where $H=\frac12g^{jk}p_jp_k$ is the Hamiltonian.
For integrability by quadratures, it suffices to find an additional
first integral functionally independent of the Hamiltonian.
Since any solution $y(x)$ of a metrisable
equation~\eqref{Eq:ODE_projective} defines an unparameterized
geodesic $(x,y(x))$, any first integral $I(x,y,y_x)$ of this
equation locally gives rise  to a first integral
$J(x,y,p_1,p_2)$ of the related Hamiltonian system, where
$(x^1,x^2)=(x,y)$.

In order to find the explicit relation between these first
integrals, we apply the Legendre transform to the Lagrangian
$L=g_{jk}\dot{x}^j\dot{x}^k/2$ associated with
system~\eqref{Eq:Ham_system}. Using the equations
$p_k=\partial L/\partial\dot{x}^k$, $k=1,2$, we find the velocities
\begin{equation}
\label{Eq:Velocities}
\dot{x}^1=\frac{g_{22}p_1-g_{12}p_2}{g_{11}g_{22}-g_{12}^2},
\qquad
\dot{x}^2=\frac{g_{11}p_2-g_{12}p_1}{g_{11}g_{22}-g_{12}^2}.
\end{equation}
Set $\delta=g_{11}g_{22}-g_{12}^2$. Recall that $\delta\neq0$
due to the nondegeneracy of the metric. On the open set where
$g_{22}p_1-g_{12}p_2\neq0$, we have $\dot{x}^1\neq0$ and
$y_x=\dot{x}^2/\dot{x}^1$. Hence
\[
J(x,y,p_1,p_2)=
I\!\left(x,y,
\frac{g_{11}p_2-g_{12}p_1}{g_{22}p_1-g_{12}p_2}\right).
\]
This first integral is homogeneous of degree zero in the momenta.
If $I$ is nonconstant, then $J$ is functionally independent of
the Hamiltonian on a regular open subset.

The system of partial differential equations given by \eqref{Eq:Geodesic_main_PDE} is overdetermined. Bryant  et al. applied the prolongation-projection method and found the compatibility conditions~\cite{bryant2009metrisability}. These conditions in the generic case are several point invariants: one of order five and two of order six.   They are very cumbersome. In addition, there exist non-generic cases that should be studied separately. The metrisability problem is explicitly solved for some families of physically relevant equations. These equations include  the Painlev\'{e} equations \cite{contatto2018metrisability}, the one-dimensional time-independent Schr\"{o}dinger equation \cite{agapov2024finite},  and  a family of autonomous equations with fractional-quadratic first integrals \cite{agapov2024integrable}.

Analyzing system \eqref{Eq:Geodesic_main_PDE}, we see that it provides the necessary and sufficient condition of equation \eqref{Eq:ODE_projective} to have the last Jacobi multiplier  $M=\left(\psi_3z^2+2\psi_2z+\psi_1\right)^{-3/2}$, where $z=y_x$. In other words, the invariant $H=\psi_3z^2+2\psi_2z+\psi_1$ satisfies the partial differential equation $\mathcal{X}_sH=(2\divr \mathcal{X}_s/3) H$, where
\begin{equation}
\label{Eq:PDE_invariant}
\mathcal{X}_s=\partial_x+z\partial_y+[A_3z^3+A_2z^2+A_1z+A_0]\partial_z,\quad \divr \mathcal{X}_s=3A_3z^2+2A_2z+A_1.
\end{equation}

Let us consider the case of Li\'{e}nard differential equations \eqref{Eq:Lienard}. Substituting $A_0=G_y$, $A_1=F_{yy}$, $A_2=0$, and $A_3=0$ into \eqref{Eq:Geodesic_main_PDE}, we conclude that the resulting system coincides with \eqref{Eq:Main_system} provided that $\psi_1=C_2$, $\psi_2=C_1$, $\psi_3=C_0$, and $\alpha=-2/3$. Thus, the method developed in the previous section can be used to solve the metrisability problem for Li\'{e}nard differential equations. Recall that our method requires solving systems of ordinary differential equations rather than partial differential equations. Note that the non-degeneracy condition $\Delta\neq0$ is exactly condition $\overline{(B)}_{-\frac23}$. In what follows, metrisability is understood locally. For real data and real solutions, the corresponding metrics are pseudo-Riemannian.

\begin{theorem}\label{T:Metrisability}
Let $F\in\mathcal A_1$ and $G\in\mathcal A_0$, and choose
representatives in $\mathcal A$.
Fix a point $p\in U_x\times S_y$. The Li\'{e}nard differential
equation~\eqref{Eq:Lienard} is metrisable in a neighborhood
of $p$ if and only if there exist $l,s,r\in\mathcal B$,
not all identically zero, such that the following conditions hold:
\begin{equation}
\label{Eq:metrisability_Cond1_new}
\begin{gathered}
(A)_{-\frac23}:\quad
\left(\partial_x-\frac{2}{3}F_{yy}\right)\Psi
+\left(2s-\frac{4l}{3}F_y-l_xy\right)G_y=0,\\
\overline{(B)}_{-\frac23}:\quad
\Delta(p)\neq0,\qquad
\Delta=l\Psi-\left(s-\frac{l_x}{2}y-\frac{2l}{3}F_y\right)^2.
\end{gathered}
\end{equation}
Here the function $\Psi$ is given by
\eqref{Eq:metrisability_Psi} with $\alpha=-2/3$, and condition
$(A)_{-\frac23}$ holds identically on $U_x\times S_y$.
On a sufficiently small neighborhood of $p$ where $\Delta$
does not vanish, the related metrics take the form
\begin{equation}
\label{Eq:Metric_metrisability}
ds^2=\frac{\Psi}{\Delta^2}(dx)^2
+\frac{6s-4lF_y-3l_xy}{3\Delta^2}dxdy
+\frac{l}{\Delta^2}(dy)^2.
\end{equation}
If, in addition, $d_{-\frac23}\geq2$, then relation
\eqref{Eq:Metric_metrisability} gives a family containing
pairwise nonproportional metrics with the same
unparametrised geodesics.
\end{theorem}
\begin{proof}
For $\alpha=-2/3$, system~\eqref{Eq:Main_system} coincides
with the metrisability system~\eqref{Eq:Geodesic_main_PDE},
for the unknowns $(C_0,C_1,C_2)$.
Theorem~\ref{T:Inv} gives the corresponding expressions
for these functions. The nondegeneracy condition at $p$
is precisely $(C_0C_2-C_1^2)(p)=\Delta(p)\neq0.$
The metric reconstruction formula then yields
\eqref{Eq:Metric_metrisability}.
\end{proof}

It is known that the solutions of system  \eqref{Eq:Geodesic_main_PDE} form a linear vector space of dimension at most six \cite{bryant2009metrisability}, see also \cite{kruglikov2008invariant}. This fact is in accordance with our results in the case of Li\'{e}nard differential equations.  The dimension of the solution space is called the   degree of mobility of the related metrisable equation given by \eqref{Eq:ODE_projective}. Using our notation, the degree of mobility is exactly the index of the equation. As in the previous section, we denote this number as $d_{-\frac23}$. According to Koenigs \cite{Koenigs01}, $d_{-\frac23}\neq 5$. If the degree of mobility is zero then the equation under consideration is not metrisable.  Consider a metrisable equation~\eqref{Eq:ODE_projective}. The related metrics always have a
Killing vector field and a linear in momenta first integral whenever $d_{-\frac23}\geq 4$, see \cite{Koenigs01, matveev2012two}.

 Now let us present some general results valid for any metrisable
equation~\eqref{Eq:ODE_projective}. All statements below are local,
and the coordinate neighborhood may be restricted when necessary.
Ratios are considered on open sets where their denominators do
not vanish. For real metrics, fractional powers of $\delta$ are
understood using the real cube root.

Let $ds^2=g_{jk}dx^jdx^k$ be a Riemannian or pseudo-Riemannian
metric. Projecting the equations of the related geodesic system
yields an equation of the form~\eqref{Eq:ODE_projective}.
As mentioned above, this equation reads as
\begin{equation}
\label{Eq:Metric_ODE}
y_{xx}=\Gamma^1_{22}(y_x)^3
+(2\Gamma^1_{12}-\Gamma^2_{22})(y_x)^2
+(\Gamma^1_{11}-2\Gamma^2_{12})y_x-\Gamma^2_{11}.
\end{equation}
If the index $d_{-\frac23}$ of equation~\eqref{Eq:Metric_ODE}
satisfies $d_{-\frac23}\geq2$, then this equation gives rise to
a family containing pairwise nonproportional metrics.
All of these metrics are projectively equivalent.
Here $d_{-\frac23}$ is the dimension of the solution space of
system~\eqref{Eq:Geodesic_main_PDE}.
Each homogeneous component of a polynomial in momenta first integral is
itself a first integral. Therefore, when studying polynomial in momenta
first integrals, we may consider their homogeneous components.

\begin{theorem}\label{T:Metrics_General}
Consider a Riemannian or pseudo-Riemannian metric
$ds^2=g_{jk}dx^jdx^k$. Let $d_{-\frac23}$ be the index of the
related equation~\eqref{Eq:Metric_ODE}. The following statements
are valid.
\begin{enumerate}
\item The metric $ds^2$ admits a nontrivial first integral linear
in  momenta if and only if $d_{-\frac23}\geq2$ and
system~\eqref{Eq:Geodesic_main_PDE} associated with
equation~\eqref{Eq:Metric_ODE} has a nonzero degenerate solution.

\item The metric $ds^2$ admits a first integral quadratic in
the momenta and functionally independent of the Hamiltonian
if and only if $d_{-\frac23}\geq2$.
\end{enumerate}
\end{theorem}

\begin{proof}
Let us prove Statement~$1$.  We set $\delta=\det(g_{jk})=g_{11}g_{22}-g_{12}^{2}\neq0$. First, assume that the metric
$ds^2$ admits a nontrivial first integral linear in  momenta
$J_1=a^jp_j$, where $a^j=a^j(x,y)$, $j=1,2$.
Then $J_2=g^{jk}p_jp_k/(a^jp_j)^2$
is a first integral homogeneous of degree zero in  momenta.
Using $p_j=(g_{j1}+g_{j2}z)\dot{x}^1$, or equivalently
\begin{equation}
\label{Eq:Metric_substitution}
\frac{p_1}{p_2}=\frac{g_{12}z+g_{11}}{g_{22}z+g_{12}}
\end{equation}
where both ratios are defined, we obtain the first integral
\begin{equation}
\label{Eq:Metric_FI1}
I=\frac{g_{22}z^2+2g_{12}z+g_{11}}
{\left(a^j(g_{j2}z+g_{1j})\right)^2}
\end{equation}
of equation~\eqref{Eq:Metric_ODE}.
Since $H_1=\delta^{-2/3}(g_{22}z^2+2g_{12}z+g_{11})$
is a nondegenerate solution of
system~\eqref{Eq:Geodesic_main_PDE}, the identity $I=H_1/H_2$
shows that $H_2=\delta^{-2/3}\left(a^j(g_{j2}z+g_{1j})\right)^2$
is another solution of this system.
The nontriviality of $J_1$ and the nondegeneracy of the metric
imply that $H_2\not\equiv0$. This solution is degenerate and is
therefore linearly independent of $H_1$. Thus,
$d_{-\frac23}\geq2$.

Now assume that $d_{-\frac23}\geq2$ and
system~\eqref{Eq:Geodesic_main_PDE} associated with
equation~\eqref{Eq:Metric_ODE} has a nonzero degenerate solution.
Locally, near a point where this solution is nonzero, and after
multiplying it by $-1$ if necessary, we may write it as
$H_2=(Az+B)^2$, where $A$ and $B$ do not vanish simultaneously.
The ratio of $H_1$ and $H_2$ gives the first integral
\begin{equation}
\label{Eq:Metric_FI2}
I=\frac{g_{22}z^2+2g_{12}z+g_{11}}
{(Az+B)^2\delta^{2/3}}
\end{equation}
of equation~\eqref{Eq:Metric_ODE}.
Substituting
\begin{equation}
\label{Eq:Metric_FI_Tr}
z=\frac{g_{11}p_2-g_{12}p_1}{g_{22}p_1-g_{12}p_2}
\end{equation}
into \eqref{Eq:Metric_FI2}, we obtain the first integral
$2H/J^2$ of the Hamiltonian system, where
\begin{equation}
\label{Eq:Metric_Linear_Integral}
J=\delta^{-2/3}
\left[(g_{12}A-g_{22}B)p_1
+(g_{12}B-g_{11}A)p_2\right].
\end{equation}
Since $H$ is a first integral, it follows that $\{J,H\}=0$
where $HJ\neq0$. The Poisson bracket $\{J,H\}$ is polynomial
in  momenta, so this identity holds for all momenta.
The nondegeneracy of the metric and the condition that $A$
and $B$ do not vanish simultaneously imply that $J$ is
nontrivial. This proves Statement~$1$.

The proof of Statement~$2$ is analogous.
Let $J_1=a^{jk}p_jp_k$, where $a^{jk}=a^{kj}$, be a quadratic
first integral functionally independent of the Hamiltonian.
Using the first integral $J_1/(2H)$ and the substitution
\eqref{Eq:Metric_substitution}, we find that
system~\eqref{Eq:Geodesic_main_PDE} has the solution
\begin{equation}
\label{Eq:Metric_Sol_Inv}
H_2=
\frac{a^{jk}(g_{2j}z+g_{1j})(g_{2k}z+g_{1k})}
{\delta^{2/3}}.
\end{equation}
This solution is linearly independent of $H_1$, since otherwise
$J_1$ would be a constant multiple of the Hamiltonian.
Consequently, $d_{-\frac23}\geq2$.

Conversely, let $H_2=Az^2+2Bz+C$ be a solution of
system~\eqref{Eq:Geodesic_main_PDE} linearly independent of
$H_1$. The ratio $H_2/H_1$ is a first integral of
equation~\eqref{Eq:Metric_ODE}. After the substitution
\eqref{Eq:Metric_FI_Tr}, this ratio becomes $J/(2H)$, where
\begin{equation}
\label{Eq:Metric_Sol_Inv2}
\begin{aligned}
J=\delta^{-4/3}\bigl\{&
(A g_{12}^2-2B g_{12}g_{22}+C g_{22}^2)p_1^2\\
&+\bigl[2B(g_{12}^2+g_{11}g_{22})
-2(A g_{11}+C g_{22})g_{12}\bigr]p_1p_2\\
&+(A g_{11}^2-2B g_{11}g_{12}+C g_{12}^2)p_2^2
\bigr\}.
\end{aligned}
\end{equation}
Thus $J$ is a quadratic first integral of the related Hamiltonian
system.

The correspondence between $H_2$ and $J$ is linear and
invertible, and it takes $H_1$ to $2H$. Hence $J$ is not a
constant multiple of the Hamiltonian. Since both $J$ and $H$
are homogeneous of degree two in  momenta, they are
functionally independent.
Indeed, a local relation $J=\Phi(H)$, together with homogeneity,
gives $H\Phi'(H)=\Phi(H)$ and hence $J=cH$.
Here we do not require the irreducibility of the quadratic
first integral.
\end{proof}

Let us note that the correspondence between solutions of
system~\eqref{Eq:Geodesic_main_PDE} and quadratic first integrals
of geodesic flows is described by Matveev
\cite{matveev2012two}. We include the proof to make the article
self-contained. In addition, we provide explicit relations
between first integrals of the geodesic system and those of
the corresponding metrisable equation.

\section{Autonomous Li\'{e}nard differential equations}\label{S:Autonomous}

This section is devoted to autonomous Li\'{e}nard differential equations, i.e. we assume that the functions $F$ and $G$ in~\eqref{Eq:Lienard} are independent of $x$. We will classify these equations with invariants  of the form \eqref{Eq:genLS}. In addition, we will provide the solution of the metrisability problem.

\begin{theorem}\label{T:Lienard_autonomous}
Let $\alpha\in\mathbb C$ and let $F\in\mathcal A_1$,
$G\in\mathcal A_0$ be independent of $x$.
The autonomous Li\'{e}nard differential equation~\eqref{Eq:Lienard}
possesses a nonconstant invariant of the form~\eqref{Eq:genLS}
with the cofactor $\lambda=-\alpha F_{yy}$ if and only if,
for suitable $x$-independent representatives of $F$ and $G$
in $\mathcal A$, there exist four complex numbers
$\sigma,c_1,c_2,c_3$ such that $c_1,c_2,c_3$ do not
simultaneously vanish and at least one of the following
conditions holds.
\begin{itemize}
\item Condition $1$:
\begin{equation}
\label{Eq:metrisability_Cond1_main_c1}
\begin{gathered}
\left\{
\left[(\alpha+2)c_1F_y+\sigma c_1y-2c_2\right]\partial_y
+2c_1(\alpha F_{yy}+\sigma)
\right\}V=0,\\
V(y)=G-\frac{\alpha(\alpha+2)}{8}(F_y)^2
-\frac{\alpha}{4}\left(\sigma y-\frac{2c_2}{c_1}\right)F_y
-\frac{\sigma}{2}F-\frac{\sigma^2}{8}y^2\\
+\frac{\sigma c_2}{2c_1}y
+\frac{c_2^2-c_1c_3}{2c_1^2},
\qquad c_1\neq0.
\end{gathered}
\end{equation}

\item Condition $2$:
\begin{equation}
\label{Eq:metrisability_Cond1_main_c2}
\begin{gathered}
2c_2G-\alpha(\alpha+1)c_2(F_y)^2
+\alpha(c_3-2\sigma c_2y)F_y
-2\sigma c_2F-\sigma^2c_2y^2+c_3\sigma y=0,
\end{gathered}
\end{equation}
where $c_1=0$ and $c_2\neq0$.

\item Condition $3$:
$F=-\sigma y^2/(2\alpha)$, $c_1=c_2=0$, and
$\alpha\sigma c_3\neq0$.

\item Condition $4$:
$2G_y(c_1F_y-c_2)-c_3=0$, $\alpha=\sigma=0$, and
$c_3\neq0$.
\end{itemize}
Under conditions $1$, $2$, and $3$, an invariant takes the form
\begin{equation}
\label{Eq:Invariant_exp_aut1}
\begin{aligned}
H(x,y,z)=e^{\sigma x}\biggl\{&
c_1z^2-\left[(\alpha+2)c_1F_y+\sigma c_1y-2c_2\right]z\\
&+\frac{(\alpha+1)(\alpha+2)c_1}{2}(F_y)^2
+(\alpha+1)(\sigma c_1y-2c_2)F_y\\
&+c_1(\sigma F-2G)+\frac{\sigma^2c_1}{2}y^2
-2\sigma c_2y+c_3
\biggr\}.
\end{aligned}
\end{equation}
Under condition $4$, an invariant reads as
\begin{equation}
\label{Eq:Invariant_exp_aut2}
\begin{gathered}
H(x,y,z)=c_1z^2+2(c_2-c_1F_y)z
+c_1(F_y)^2-2c_2F_y-2c_1G+c_3x.
\end{gathered}
\end{equation}
\end{theorem}

\begin{proof}
Let us prove the necessity. We use the approach of Matveev and
Shevchishin~\cite{matveev2011two}, adapting it to the present
setting. Choose $x$-independent representatives of $F$ and $G$
in $\mathcal A$. Let $\mathcal H_\alpha$ denote the vector space
of invariants of the form~\eqref{Eq:genLS} with the cofactor
$\lambda=-\alpha F_{yy}$. By assumption, this space contains
a nonconstant invariant. The dimension of $\mathcal H_\alpha$
is $d_\alpha$. As shown in Section \ref{S:Invariants},
$d_\alpha\leq6$.

Since the coefficients of equation~\eqref{Eq:Main_PDE}
do not explicitly depend on $x$, the function $\partial_xH$
also belongs to $\mathcal H_\alpha$ whenever
$H\in\mathcal H_\alpha$. Thus,
$T:H\mapsto\partial_xH$ is a linear endomorphism of the
finite-dimensional complex vector space $\mathcal H_\alpha$.

First, consider the case where $T$ has a nonconstant
eigenfunction $H$ with eigenvalue $\sigma\in\mathbb C$.
The relation $\partial_xH=\sigma H$ implies $
H(x,y,z)=e^{\sigma x}W(y,z)$,
where $W(y,z)$ is a polynomial in $z$ of degree at most two.
Using Theorem~\ref{T:Inv} and comparing the coefficients
in \eqref{Eq:metrisability_Cond1_fg}, we find $l(x)=c_1\exp[\sigma x]$, $s(x)=c_2\exp[\sigma x]$,  $r(x)=c_3\exp[\sigma x]$,
where $c_1,c_2,c_3\in\mathbb C$ do not simultaneously vanish.
Substituting these expressions into condition $(A)_\alpha$
gives
\begin{equation}
\label{Eq:autonomous_compatibility}
\begin{aligned}
&2\left[(\alpha+2)c_1F_y+\sigma c_1y-2c_2\right]G_y
+4c_1(\alpha F_{yy}+\sigma)G\\
&\quad=
\Bigl\{
\alpha(\alpha+1)(\alpha+2)c_1(F_y)^2
+2\alpha(\alpha+1)(\sigma c_1y-2c_2)F_y\\
&\qquad\qquad
+2\alpha\sigma c_1F
+\alpha(\sigma^2c_1y^2-4\sigma c_2y+2c_3)
\Bigr\}F_{yy}\\
&\qquad
+(\alpha+1)(\alpha+2)\sigma c_1(F_y)^2
+2(\alpha+1)\sigma(\sigma c_1y-2c_2)F_y\\
&\qquad
+2\sigma^2c_1F
+\sigma(\sigma^2c_1y^2-4\sigma c_2y+2c_3).
\end{aligned}
\end{equation}

Assuming that $c_1\neq0$, we introduce the function $V(y)$
given in \eqref{Eq:metrisability_Cond1_main_c1}.
Then relation~\eqref{Eq:autonomous_compatibility} takes the
form of condition~$1$.

If $c_1=0$ and $c_2\neq0$, we consider
\eqref{Eq:autonomous_compatibility} as a differential equation
for $G$. Integrating this equation gives
\begin{equation}
\label{Eq:autonomous_compatibility_new1}
2c_2G-\alpha(\alpha+1)c_2(F_y)^2
+\alpha(c_3-2\sigma c_2y)F_y
-2\sigma c_2F-\sigma^2c_2y^2+c_3\sigma y=k_0,
\end{equation}
where $k_0\in\mathbb C$.
Replacing $G$ by $G-k_0/(2c_2)$, which represents the same
class in $\mathcal A_0$, we obtain condition~$2$.

Finally, if $c_1=c_2=0$ and $c_3\neq0$, relation
\eqref{Eq:autonomous_compatibility} reduces to $\alpha F_{yy}+\sigma=0$.
The corresponding invariant is $H=c_3\exp(\sigma x)$.
Since it is nonconstant, $\sigma\neq0$, and hence
$\alpha\neq0$. Integrating the equation for $F$ gives
$F=-\sigma y^2/(2\alpha)+ay+b$, $a,b\in\mathbb C$.
Choosing the representative with $a=b=0$ gives condition~$3$.
In all three cases, the invariant takes the form
\eqref{Eq:Invariant_exp_aut1}.

Now consider the case where all eigenfunctions of $T$ are
constant. Then the only eigenvalue of $T$ is zero.
The existence of a nonzero constant invariant implies
$\alpha F_{yy}\equiv0$.
If $F_{yy}\equiv0$, then $H_*=z^2-2G(y)$
is a nonconstant invariant satisfying $TH_*=0$, contrary
to our assumption. Therefore, $F_{yy}\not\equiv0$ and
$\alpha=0$.

The operator $T$ is nilpotent, and its kernel consists of
the constant functions. Since $\mathcal H_0$ contains a
nonconstant invariant, a Jordan chain for the eigenvalue
zero contains an element $H$ such that $TH=c_3$, where  $c_3\in\mathbb C\setminus\{0\}$.
Consequently, $H=c_3x+W(y,z)$,
where $W$ is a polynomial in $z$ of degree at most two.
Using Theorem~\ref{T:Inv}, we obtain $l=c_1$, $s=c_2$, $r=c_3x+r_0$,
where $c_1,c_2,r_0\in\mathbb C$.
Subtracting the constant invariant $r_0$, we may assume
that $r=c_3x$.
Substituting these expressions into condition $(A)_0$
gives $2G_y(c_1F_y-c_2)-c_3=0$.
Thus, condition~$4$ holds, with $\sigma=0$, and the
corresponding invariant is given by
\eqref{Eq:Invariant_exp_aut2}. This invariant is a first
integral of the related equation.

The sufficiency is verified by substituting
\eqref{Eq:Invariant_exp_aut1} and
\eqref{Eq:Invariant_exp_aut2}, together with the corresponding
conditions, into \eqref{Eq:Main_PDE}.
The restrictions in each case ensure that the resulting
invariant is non-constant.
\end{proof}

\begin{remark}
Set $D(y)=(\alpha+2)c_1F_y+\sigma c_1y-2c_2$.
Relation
\eqref{Eq:metrisability_Cond1_main_c1} is a first-order linear
differential equation for $V(y)$ whenever $D(y)\not\equiv0$. Integrating this equation
yields
\begin{equation}
\label{Eq:V}
V(y)=v_0\exp\left[
-2c_1\partial_y^{-1}
\left(
\frac{\alpha F_{yy}+\sigma}
{(\alpha+2)c_1F_y+\sigma c_1y-2c_2}
\right)
\right],
\end{equation}
where $v_0\in\mathbb C$ is the constant of integration.
When this formula is used to construct $G$, the condition
$V\in\mathcal A$, equivalently $G\in\mathcal A$, must also
be satisfied.

If $D(y)\equiv0$, condition~$1$ reduces to $2c_1(\alpha F_{yy}+\sigma)V=0$.
Thus, $V\equiv0$ if $\alpha F_{yy}+\sigma\not\equiv0$.
If $\alpha F_{yy}+\sigma\equiv0$, this equation imposes
no further restriction on $V\in\mathcal A$.
\end{remark}

Suppose that a Li\'{e}nard differential equation is metrisable.
It follows from Theorem~\ref{T:Metrics_General} that the related
metrics have a nontrivial first integral linear in  momenta
whenever system $(S)_{-\frac23}$ possesses a nonzero degenerate
solution. Let us find the functions $F(y)\in\mathcal{A}_1$ and
$G(y)\in\mathcal{A}_0$ such that the autonomous Li\'{e}nard
differential equation~\eqref{Eq:Lienard} has a nonzero degenerate
invariant of degree at most two in $z$, i.e. an invariant satisfying
condition $(B)_{-\frac23}$. Thus, we will consider the existence of the invariant
\begin{equation}
\label{Eq:Inv_Deg1_General}
H(x,y,z)=B_0(x,y)z+B_1(x,y)
\end{equation}
with the cofactor $\lambda=F_{yy}/3$. Conversely, the square of any
such invariant is a degenerate invariant with the cofactor
$\lambda=2F_{yy}/3$.

\begin{lemma}\label{L:Lienard_degenerate}
Let $F\in\mathcal{A}_1$ and $G\in\mathcal{A}_0$ be independent of
$x$. The autonomous Li\'{e}nard differential
equation~\eqref{Eq:Lienard} possesses a nonzero invariant
\eqref{Eq:Inv_Deg1_General} with the cofactor $\lambda=F_{yy}/3$
if and only if one of the following conditions holds for suitable
representatives of $F$ and $G$ in $\mathcal A$ that are independent
of $x$.
\begin{enumerate}

\item $F(y)=a_1y^2$ and $G_{yyy}\not\equiv0$, where
$a_1\in\mathbb C$. A related invariant takes the form
$H=\exp(2a_1x/3)$.

\item
\[
(F(y),G(y))=
\left(a_0y^3+a_1y^2,
-a_0^2y^4-\frac43a_0a_1y^3+a_2y^2+a_3y\right),
\]
where $a_0,\ldots,a_3\in\mathbb C$. The related invariants take
the form
\begin{equation}
\label{Eq:Inv_Deg1}
H=uz-2a_0uy^2-\left(u_x+\frac{4a_1}{3}u\right)y+v,
\end{equation}
where $u,v\in\mathcal B$ are not both identically zero.
If $a_0=0$, then the functions $u(x)$ and $v(x)$ satisfy the
system of equations
\begin{equation}
\label{Eq:Inv_Deg1_syst}
\begin{gathered}
9u_{xx}+6a_1u_x-2\{4a_1^2+9a_2\}u=0,\quad
3v_x-2a_1v+3a_3u=0.
\end{gathered}
\end{equation}
If $a_0\neq0$, then the function $v(x)$ is given by the relation
\begin{equation}
\label{Eq:Inv_Deg1_v}
v=-\frac{1}{2a_0}u_{xx}-\frac{a_1}{3a_0}u_x
+\frac{4a_1^2+9a_2}{9a_0}u,
\end{equation}
and the function $u(x)$ satisfies the equation
\begin{equation}
\label{Eq:Inv_Deg1_u}
u_{xxx}-\frac{2\{2a_1^2+3a_2\}}{3}u_x
+\frac{2\{8a_1^3-27a_0a_3+18a_1a_2\}}{27}u=0.
\end{equation}

\item
\[
G(y)=-\frac19(F_y)^2+\frac13(v_0-\delta y)F_y
+\delta F+\frac{\delta^2}{2}y^2-\delta v_0y,
\qquad F_{yyyy}\not\equiv0,
\]
where $\delta,v_0\in\mathbb C$. A related invariant takes the form
\begin{equation}
\label{Eq:Inv_Deg2_Inv}
H=\left(z-\frac23F_y-\delta y+v_0\right)\exp(\delta x).
\end{equation}

\end{enumerate}
\end{lemma}

\begin{proof}
Substituting expression~\eqref{Eq:Inv_Deg1_General} into the
linear partial differential equation
\begin{equation}
\label{Eq:Main_PDE_degenerate}
H_x+zH_y+(F_{yy}z+G_y)H_z-\frac13F_{yy}H=0
\end{equation}
and setting to zero the coefficients of different powers of $z$,
we get the system
\begin{equation}
\label{Eq:Main_system_degenerate}
\begin{gathered}
(B_0)_y=0,\quad
(B_1)_y+(B_0)_x+\frac23F_{yy}B_0=0,\quad
(B_1)_x-\frac13F_{yy}B_1+G_yB_0=0.
\end{gathered}
\end{equation}
Integrating the first two equations of this system yields
$B_0=u(x)$ and $B_1=v-2uF_y/3-u_xy$, where $u,v\in\mathcal B$.
Substituting these expressions into the third equation and
integrating with respect to $y$, we obtain the condition
\begin{equation}
\label{Eq:degenerate_compatibility}
uG+\frac19u(F_y)^2+\frac13(u_xy-v)F_y
-\frac12u_{xx}y^2-u_xF+v_xy+w=0,
\end{equation}
where $w\in\mathcal B$.

First, we assume that $u\equiv0$. Since the invariant is nonzero,
$v\not\equiv0$. The third equation of
\eqref{Eq:Main_system_degenerate} gives $3v_x=F_{yy}v$.
Separating the variables, we conclude that $F_{yy}=2a_1$,
where $a_1\in\mathbb C$. Choosing the representative $F=a_1y^2$,
we get $v=v_{0}\exp(2a_1x/3)$, where $v_{0}\in\mathbb{C}\setminus\{0\}$, and $w=0$.
Multiplying the invariant by a nonzero constant, we obtain
$H=\exp(2a_1x/3)$. This gives item~$1$ if
$G_{yyy}\not\equiv0$. Otherwise, the equation is included in
item~$2$ with $a_0=0$.

Now let us consider the case $u\not\equiv0$. We divide
\eqref{Eq:degenerate_compatibility} by $u$ and differentiate the
result once with respect to $x$ and three times with respect to
$y$. We obtain
\begin{equation}
\label{Eq:Main_system_degenerate_diff}
\left\{
\partial_x\left(\frac{u_x}{u}\right)y
-\partial_x\left(\frac vu\right)
\right\}F_{yyyy}=0.
\end{equation}
Thus, either $F_{yyyy}\equiv0$ or the expression in brackets
vanishes identically.

If $F_{yyyy}\equiv0$, then we can choose the representative
$F=a_0y^3+a_1y^2$. Substituting this expression into
\eqref{Eq:degenerate_compatibility}, we see that $G(y)$ is a
polynomial with fixed coefficients of $y^4$ and $y^3$.
Since $G$ is independent of $x$, it takes the form
\[
G=-a_0^2y^4-\frac43a_0a_1y^3+a_2y^2+a_3y+g_0,
\qquad a_2,a_3,g_0\in\mathbb C.
\]
The constant term in \eqref{Eq:degenerate_compatibility} gives
$w=-g_0u$. Replacing $G$ by $G-g_0$, which represents the same
element of $\mathcal A_0$, we therefore set $w=0$. Condition
\eqref{Eq:degenerate_compatibility} now reads as
\[
\left[
\left(a_2+\frac{4a_1^2}{9}\right)u
-\frac{a_1}{3}u_x-\frac12u_{xx}-a_0v
\right]y^2
+\left(a_3u-\frac{2a_1}{3}v+v_x\right)y=0.
\]
Equating to zero the coefficients of $y^2$ and $y$, we obtain
system~\eqref{Eq:Inv_Deg1_syst} if $a_0=0$. If $a_0\neq0$,
the first equation gives \eqref{Eq:Inv_Deg1_v}; substituting
this expression into the second equation gives
\eqref{Eq:Inv_Deg1_u}. This proves item~$2$.
All solutions of these equations are holomorphic in $x$.
Note that system~\eqref{Eq:Inv_Deg1_syst} admits the solution
$u=0$ and $v=v_{0}\exp(2a_1x/3)$ found previously.

Finally, suppose that $F_{yyyy}\not\equiv0$. Relation
\eqref{Eq:Main_system_degenerate_diff} implies
\begin{equation}
\label{Eq:Main_system_degenerate_diff_new}
\partial_x\left(\frac{u_x}{u}\right)=0,\qquad
\partial_x\left(\frac vu\right)=0.
\end{equation}
Solving this system yields
$u=u_0\exp(\delta x)$ and $v=v_0u$, where
$u_0\in\mathbb C\setminus\{0\}$ and $\delta,v_0\in\mathbb C$.
After substitution into \eqref{Eq:degenerate_compatibility},
all terms in this equation divided by $u$, except possibly
$w/u$, depend only on $y$. Consequently, $w/u=k$ is constant.
Replacing $G$ by $G+k$, we set $w=0$. Multiplying the invariant
by $u_0^{-1}$, we can also set $u_0=1$. We then obtain
\[
G=-\frac19(F_y)^2+\frac13(v_0-\delta y)F_y
+\delta F+\frac{\delta^2}{2}y^2-\delta v_0y
\]
and invariant~\eqref{Eq:Inv_Deg2_Inv}. This gives item~$3$.

The sufficiency follows by substituting these invariants
and the corresponding conditions into
\eqref{Eq:Main_PDE_degenerate}.

The space of invariants of the form
\eqref{Eq:Inv_Deg1_General}, including the zero invariant, is
three-dimensional in item~$2$. Indeed, if $a_0=0$,
system~\eqref{Eq:Inv_Deg1_syst} has three independent initial
constants. If $a_0\neq0$, equation~\eqref{Eq:Inv_Deg1_u}
is of order three, and $v$ is uniquely determined by $u$.

In item~$1$, every invariant has $u\equiv0$, since otherwise
the preceding calculation would imply $G_{yyy}\equiv0$.
Thus, the space of invariants is one-dimensional.

The same conclusion holds in item~$3$. To verify this,
suppose that two pairs $(\delta_1,v_{01})$ and
$(\delta_2,v_{02})$ give the same $G$ modulo an additive
constant for a fixed representative of $F$.
Subtracting the corresponding expressions for $G$ and
differentiating three times with respect to $y$, we obtain
\[
\bigl[(\delta_1-\delta_2)y-(v_{01}-v_{02})\bigr]
F_{yyyy}=0.
\]
Since $F_{yyyy}\not\equiv0$, we get
$\delta_1=\delta_2$ and $v_{01}=v_{02}$. Hence every nonzero
invariant in item~$3$ is a constant multiple of
\eqref{Eq:Inv_Deg2_Inv}.
\end{proof}

\begin{remark}
Consider an autonomous Li\'{e}nard differential
equation~\eqref{Eq:Lienard} satisfying the conditions of item~$2$.
If $a_0\neq0$, choose three linearly independent solutions
$u_j$ of equation~\eqref{Eq:Inv_Deg1_u} and define $v_j$
by \eqref{Eq:Inv_Deg1_v}. If $a_0=0$, choose three linearly
independent solutions $(u_j,v_j)$ of
system~\eqref{Eq:Inv_Deg1_syst}. Let $H_j$, $j=1,2,3$,
be the corresponding invariants given by
\eqref{Eq:Inv_Deg1}. Their ratios $I_1=H_1/H_3$ and $I_2=H_2/H_3$
are functionally independent first integrals of the related equation.
\end{remark}

Suppose an autonomous Li\'{e}nard differential equation~\eqref{Eq:Lienard} satisfies the conditions of Lemma~\ref{L:Lienard_degenerate} and
is metrisable. By Lemma \ref{L:Lienard_degenerate} and Theorem \ref{T:Metrics_General}, the related metrics always have a linear first integral. Let us find all metrisable autonomous Li\'{e}nard differential equations in the leftover cases.

\begin{theorem}\label{T:Lienard_autonomous_metrisability}
Let $F\in\mathcal A_1$ and $G\in\mathcal A_0$ be independent
of $x$, and suppose that none of the conditions of
Lemma~\ref{L:Lienard_degenerate} holds.
The autonomous Li\'{e}nard differential
equation~\eqref{Eq:Lienard} is metrisable if and only if, for suitable
$x$-independent representatives of $F$ and $G$ in $\mathcal A$,
there exist constants
$\sigma\in\mathbb C\setminus\{0\}$ and
$c_1,c_2,c_3\in\mathbb C$ such that one of the following
conditions holds.
\begin{enumerate}

\item $c_1\neq0$ and
\begin{equation}
\label{Eq:autonomous_metrisability_case1}
\left\{
\left(\frac43c_1F_y+\sigma c_1y-2c_2\right)\partial_y
+2c_1\left(\sigma-\frac23F_{yy}\right)
\right\}V=0,
\qquad V\not\equiv0,
\end{equation}
where
\begin{equation}
\label{Eq:autonomous_metrisability_V}
\begin{gathered}
V(y)=G+\frac19(F_y)^2
+\frac16\left(\sigma y-\frac{2c_2}{c_1}\right)F_y
-\frac{\sigma}{2}F-\frac{\sigma^2}{8}y^2
+\frac{\sigma c_2}{2c_1}y
+\frac{c_2^2-c_1c_3}{2c_1^2}.
\end{gathered}
\end{equation}

\item $c_1=0$, $c_2\neq0$, and
\begin{equation}
\label{Eq:autonomous_metrisability_case2}
\begin{gathered}
G=-\frac19(F_y)^2
+\left(\frac{c_3}{3c_2}-\frac{2\sigma}{3}y\right)F_y
+\sigma F+\frac{\sigma^2}{2}y^2
-\frac{\sigma c_3}{2c_2}y.
\end{gathered}
\end{equation}

\end{enumerate}
Set
\begin{equation}
\label{Eq:autonomous_metrisability_DQ}
\begin{gathered}
D(y)=\frac43c_1F_y+\sigma c_1y-2c_2,\\
Q(y)=\frac29c_1(F_y)^2
+\frac13(\sigma c_1y-2c_2)F_y
+c_1(\sigma F-2G)
+\frac{\sigma^2c_1}{2}y^2-2\sigma c_2y+c_3.
\end{gathered}
\end{equation}
A related invariant with the cofactor $\lambda=2F_{yy}/3$
takes the form
\begin{equation}
\label{Eq:autonomous_metrisability_invariant}
H=\exp(\sigma x)\left(c_1z^2-Dz+Q\right).
\end{equation}
Define
\begin{equation}
\label{Eq:autonomous_metrisability_delta}
\Delta_0(y)=c_1Q-\frac14D^2
=
\begin{cases}
-2c_1^2V,&\text{in item }1,\\
-c_2^2,&\text{in item }2.
\end{cases}
\end{equation}
On a neighborhood where $\Delta_0\neq0$, a related metric
is given by
\begin{equation}
\label{Eq:autonomous_metrisability_metric}
ds^2=\frac{e^{-3\sigma x}}{\Delta_0^2}
\left\{Q(dx)^2-D\,dx\,dy+c_1(dy)^2\right\}.
\end{equation}
In particular, the metric in item~$1$ is defined where
$V\neq0$, whereas the metric in item~$2$ is non-degenerate
throughout its domain of definition.
None of the metrics of this theorem admits a
nontrivial Killing vector field and a linear in momenta first integral.
\end{theorem}

\begin{proof}
As proved in Lemma~\ref{L:Lienard_degenerate}, the equation under consideration does not have
nonzero invariants of degree at most one in $z$ with the
cofactor $\lambda=F_{yy}/3$. Therefore, it has no nonzero
degenerate invariants of degree at most two in $z$ with the
cofactor $\lambda=2F_{yy}/3$. We also note that $F_{yy}$ is not constant. Indeed, if
$F_{yy}=\mu\in\mathbb C$, then
$H=\exp(\mu x/3)$ would be a nonzero invariant with the
cofactor $\lambda=F_{yy}/3$, contrary to our assumption.

Let us prove the necessity. Suppose that the equation under consideration is metrisable.
Theorem~\ref{T:Metrisability} gives a non-degenerate invariant
of the form~\eqref{Eq:genLS} with the cofactor
$\lambda=2F_{yy}/3$. This invariant is nonconstant because
$F_{yy}\not\equiv0$. We apply
Theorem~\ref{T:Lienard_autonomous} with $\alpha=-2/3$.
Condition~$4$ of that theorem is impossible because
$\alpha\neq0$, and condition~$3$ is impossible because it
requires $F_{yy}$ to be constant. Consequently, either
condition~$1$ or condition~$2$ holds.

Substituting $\alpha=-2/3$ into these two conditions gives
\eqref{Eq:autonomous_metrisability_case1} without the
restriction $V\not\equiv0$, or
\eqref{Eq:autonomous_metrisability_case2}, respectively.
The corresponding invariant is
\eqref{Eq:autonomous_metrisability_invariant}. Its coefficients
satisfy $C_0=c_1\exp(\sigma x)$, $C_1=-\frac12D\exp(\sigma x)$, $C_2=Q\exp(\sigma x)$,
and hence $\Delta=C_0C_2-C_1^2=\Delta_0\exp(2\sigma x)$.
A direct calculation gives
\eqref{Eq:autonomous_metrisability_delta}. Since every nonzero
invariant under consideration is non-degenerate, we obtain
$V\not\equiv0$ in item~$1$. In item~$2$, the condition
$c_2\neq0$ already ensures non-degeneracy.

It remains to show that $\sigma\neq0$. Suppose that
$\sigma=0$. In item~$1$, condition
\eqref{Eq:autonomous_metrisability_case1} becomes
\begin{equation}
DV_y-D_yV=0,\qquad D=\frac43c_1F_y-2c_2.
\end{equation}
Since $F_{yy}$ is not constant, $D\not\equiv0$, and therefore
$V=kD$ for some $k\in\mathbb C$. Substitution into
\eqref{Eq:autonomous_metrisability_V} gives
\begin{equation}
G=-\frac19(F_y)^2+\frac{v_0}{3}F_y+g_0
\end{equation}
for some $v_0,g_0\in\mathbb C$. The same relation follows
from \eqref{Eq:autonomous_metrisability_case2} when
$\sigma=0$, with $v_0=c_3/c_2$ and $g_0=0$.
In either case, $K=z-2F_y/3+v_0$
is a nonzero invariant with the cofactor $\lambda=F_{yy}/3$.
By Lemma~\ref{L:Lienard_degenerate}, this contradicts our
assumption. Thus, $\sigma\neq0$.

Now let us prove the sufficiency. Suppose that one of the
two stated conditions holds. These are precisely
conditions~$1$ and~$2$ of
Theorem~\ref{T:Lienard_autonomous} with $\alpha=-2/3$.
Therefore, \eqref{Eq:autonomous_metrisability_invariant}
is an invariant with the cofactor $\lambda=2F_{yy}/3$.
The corresponding functions $l=c_1\exp(\sigma x)$, $s=c_2\exp(\sigma x)$, and $r=c_3\exp(\sigma x)$
belong to $\mathcal B$.

In item~$1$, $V\not\equiv0$, so there is a nonempty open
set on which $\Delta_0\neq0$. In item~$2$,
$\Delta_0=-c_2^2\neq0$. Thus,
Theorem~\ref{T:Metrisability} gives a metric on a
neighborhood of every point where $\Delta_0\neq0$.
Substituting the  coefficients $C_0,C_1,C_2$
and $\Delta=\exp(2\sigma x)\Delta_0$ into
\eqref{Eq:Metric_metrisability}, we obtain
\eqref{Eq:autonomous_metrisability_metric}. Its determinant is
$\det(g)=\exp(-6\sigma x)/\Delta_0^3\neq0$.

Finally, if a metric related to the equation admitted a
nontrivial  Killing vector field, then
Theorem~\ref{T:Metrics_General} would give a nonzero
degenerate solution of system $(S)_{-\frac23}$.
This contradicts the first paragraph of the proof.
\end{proof}

\begin{remark}
In item~$1$, the function $D$ cannot vanish identically,
since this would imply that $F_{yy}$ is constant.
Since $D_y=c_1\left(4F_{yy}/3+\sigma\right)$,
equation~\eqref{Eq:autonomous_metrisability_case1} can be
written as $DV_y+(3c_1\sigma-D_y)V=0$.
Suppose that $D\not\equiv0$, then
its nonzero solutions take the form
\begin{equation}
\label{Eq:autonomous_metrisability_V_integrated}
V(y)=v_0D(y)\exp\left[
-3c_1\sigma\partial_y^{-1}\left(\frac{1}{D(y)}\right)
\right],
\qquad v_0\in\mathbb C\setminus\{0\}.
\end{equation}
Together with \eqref{Eq:autonomous_metrisability_V},
this determines $G$ by one quadrature.
When these formulae are used to construct $G$,
the condition $V\in\mathcal A$, equivalently $G\in\mathcal A$,
must also be satisfied.
\end{remark}

It is straightforward to solve the metrisability problem for Li\'{e}nard differential equations satisfying the conditions of Lemma \ref{L:Lienard_degenerate}. In this article we will not present this result. Recall that the related metrics always have a linear in momenta first integral.

Suppose we know some integrable family of autonomous Li\'{e}nard differential equations. If  equations of this family satisfy the conditions of Theorem \ref{T:Lienard_autonomous_metrisability}, then it is straightforward to find integrable metrics without linear in momenta first integrals. Let $f$ and $g$ be independent of $x$. The following relation
\begin{equation}\label{Eq:Chiellini1}
     \partial_y\left(\frac{g(y)}{f(y)}\right)=2\beta f(y),\quad \beta\in\mathbb{C}
    \end{equation}
    is known as  the Chiellini integrability condition \cite{Chiellini01}. Explicit first integrals  of Li\'{e}nard differential equations satisfying this condition are given in \cite{harko2014class}, see also \cite{demina2025integrability}. Choosing suitable  $x$-independent representatives $F,G\in\mathcal A$, we rewrite condition \eqref{Eq:Chiellini1} as   $G=\beta(F_y)^2$ whenever $\beta\neq0 $ and as  $G=\kappa F_y$ otherwise. Here $\kappa\in\mathbb{C}\setminus\{0\}$. The case $\kappa=0$ we exclude as a trivial one. Let us find integrable metrics produced by autonomous Li\'{e}nard differential equations \eqref{Eq:Lienard} satisfying the  Chiellini integrability condition. The cases $\beta\neq0$ and $\beta=0$ will be considered separately.

\begin{theorem}\label{T:chiellini-metrisability}
Let $G=\beta(F_y)^2$, where $G$ and $F\in\mathcal A$ are independent of $x$. Suppose that the autonomous Li\'enard differential equation \eqref{Eq:Lienard} is metrisable and the related metrics do not have nontrivial first integrals
linear in the momenta. Then $F_{yyy}\not\equiv 0$, $\beta\in\mathbb C\setminus\left\{0,-1/9\right\}$, and
after a suitable choice of the additive constant in $F$, there exist constants $\sigma\in\mathbb C\setminus\{0\}$, $c_1,c_2,c_3\in\mathbb{C}$
such that at least one of the following conditions holds.
\begin{enumerate}
    \item\label{item:chiellini-one}
    $c_1\neq0$, and $F$ satisfies the second-order ODE
    \begin{equation}\label{eq:chiellini-case-one}
    \begin{aligned}
        &\Bigg[
        4\left(\beta+\frac19\right)c_1(F_y)^2
        +6\left(\beta+\frac19\right)
            (\sigma c_1y-2c_2)F_y\\
        &\qquad +2\sigma c_1F+\sigma^2c_1y^2
        -4\sigma c_2y+2c_3
        \Bigg]F_{yy}\\
        &\quad +\sigma\Bigg[
        \left(6\beta-\frac23\right)c_1(F_y)^2
        -(\sigma c_1y-2c_2)F_y\\
        &\qquad -3\sigma c_1F-\frac32\sigma^2c_1y^2
        +6\sigma c_2y-3c_3
        \Bigg]=0
    \end{aligned}
    \end{equation}
    with $V\not\equiv0$, where
    \begin{equation}\label{eq:chiellini-V}
            V(y)=        \left(\beta+\frac19\right)(F_y)^2
        +\frac16\left(\sigma y-\frac{2c_2}{c_1}\right)F_y
        -\frac{\sigma}{2}F
        -\frac{\sigma^2}{8}y^2
        +\frac{\sigma c_2}{2c_1}y
        +\frac{c_2^2-c_1c_3}{2c_1^2}.
        \end{equation}

    \item\label{item:chiellini-two}
    $c_1=0$, $c_2\ne0$, and $F$ satisfies the first-order ODE
    \begin{equation}\label{eq:chiellini-case-two}
           c_2\left(\beta+\frac19\right)(F_y)^2
        +\frac13(2\sigma c_2y-c_3)F_y
                     -\sigma c_2F-\frac{\sigma^2c_2}{2}y^2
        +\frac{\sigma c_3}{2}y=0.
       \end{equation}
\end{enumerate}
Set
\begin{equation}\label{eq:chiellini-D-Q}
\begin{gathered}
    D(y)=\frac43c_1F_y+\sigma c_1y-2c_2,
    \\
    Q(y)=2\left(\frac19-\beta\right)c_1(F_y)^2
    +\frac13(\sigma c_1y-2c_2)F_y
    +c_1\sigma F         +\frac{\sigma^2c_1}{2}y^2-2\sigma c_2y+c_3,
\end{gathered}
\end{equation}
and
\begin{equation}\label{eq:chiellini-Delta}
    \Delta_0(y)=c_1Q-\frac14D^2
    =
    \begin{cases}
        -2c_1^2V, & \text{in item~\ref{item:chiellini-one}},\\
        -c_2^2,   & \text{in item~\ref{item:chiellini-two}}.
    \end{cases}
\end{equation}
On a neighborhood where $\Delta_0\ne0$, a related metric is given by
\begin{equation}\label{eq:chiellini-metric}
    ds^2=
    \frac{e^{-3\sigma x}}{\Delta_0^2}
    \left\{
        Q\,(dx)^2-D\,dx\,dy+c_1\,(dy)^2
    \right\}.
\end{equation}
The geodesic flow of metric \eqref{eq:chiellini-metric}
has the Hamiltonian
\begin{equation}\label{eq:chiellini-Hamiltonian}
    H=
    \frac12e^{3\sigma x}\Delta_0
    \left(
        c_1p_1^2+Dp_1p_2+Qp_2^2
    \right)
\end{equation}
and an additional first integral defined as follows.
Choose $\varepsilon\in\mathbb C\setminus\{0,-1\}$ satisfying
$\beta=(\varepsilon+1)/(2\varepsilon^2)$.
If $\varepsilon\neq-2$, then
\begin{equation}\label{eq:chiellini-first-integral}
    J=
    \left(
        z-\frac{\varepsilon+1}{\varepsilon}F_y
    \right)^{\varepsilon+1}
    \left(
        z+\frac1{\varepsilon}F_y
    \right),\quad  z=\frac{Dp_1+2Qp_2}{2c_1p_1+Dp_2}.
\end{equation}
If $\varepsilon=-2$, equivalently $\beta=-1/8$, then
\begin{equation}\label{eq:chiellini-log-integral}
    J=
    \left(z-\frac12F_y\right)\exp\left(-\frac{F_y}{2z-F_y}\right), \quad  z=\frac{Dp_1+2Qp_2}{2c_1p_1+Dp_2}.
\end{equation}
These expressions are understood locally on domains where their denominators do not vanish, after choosing branches of the powers. The first integral $J$ is homogeneous of degree zero in the momenta and is functionally independent of the Hamiltonian on a regular open subset.

Conversely, suppose that
$G=\beta(F_y)^2$, $F_{yyy}\not\equiv0$, and
$\beta\in\mathbb C\setminus\{0,-1/9\}$.
If there exist constants
$\sigma\in\mathbb C\setminus\{0\}$ and
$c_1,c_2,c_3\in\mathbb C$
such that all the conditions of either
item~\ref{item:chiellini-one} or
item~\ref{item:chiellini-two} hold, then the autonomous
Li\'enard differential equation \eqref{Eq:Lienard}
is metrisable on a nonempty open subset.
A related metric is given by \eqref{eq:chiellini-metric}
and admits no nontrivial first integral linear in momenta.
\end{theorem}
\begin{proof}
Conditions \ref{item:chiellini-one} and \ref{item:chiellini-two} follow from those given in Theorem \ref{T:Lienard_autonomous_metrisability}. The related metrics are also given in Theorem \ref{T:Lienard_autonomous_metrisability}. Explicit expressions of first integrals \eqref{eq:chiellini-first-integral} and \eqref{eq:chiellini-log-integral} are presented in \cite{demina2025integrability}.  Let us exclude the existence of linear in momenta first integrals. For this aim we use Lemma \ref{L:Lienard_degenerate}.

First, the necessity of $F_{yyy}\not\equiv0$ follows from
Lemma \ref{L:Lienard_degenerate}.
The value $\beta=-1/9$ is also excluded, since
$z-2F_y/3$ is then an invariant with the cofactor $\lambda=F_{yy}/3$.
In either item of the theorem, the functions $D$ and $Q$ satisfy
\begin{equation}
    Q_y=\left(\sigma+\frac13F_{yy}\right)D
          -4\beta c_1F_yF_{yy},
    \qquad
    (2F_{yy}-3\sigma)Q+6\beta F_yF_{yy}D=0.
\end{equation}
If $\beta=0$, the second identity and
$F_{yyy}\not\equiv0$ imply $Q\equiv0$.
The first identity then gives $D\equiv0$, contradicting
$\Delta_0\not\equiv0$.
Thus $\beta\notin\{0,-1/9\}$.

We now prove the assertion concerning linear first integrals
in the converse direction.
The first alternative of Lemma \ref{L:Lienard_degenerate}
is impossible because $F_{yyy}\not\equiv0$.
In its second alternative, either $F_{yyy}\equiv0$ or
comparison of the coefficients of $y^4$ in
$G=\beta(F_y)^2$ gives $\beta=-1/9$.
Hence this alternative is also excluded.

Suppose that the third alternative of
Lemma \ref{L:Lienard_degenerate} holds.
For the representatives used here, its condition reads
\begin{equation}
    \beta(F_y)^2
    =-\frac19(F_y)^2+\frac13(v_0-\delta y)F_y
      +\delta F+\frac{\delta^2}{2}y^2-\delta v_0y+g_0,
\end{equation}
where $\delta,v_0,g_0\in\mathbb C$.
Differentiating, we obtain
\begin{equation}\label{Eq:New_add1}
    \left[
        6\left(\beta+\frac19\right)F_y+\delta y-v_0
    \right]F_{yy}
    =\delta\left(2F_y+3\delta y-3v_0\right).
\end{equation}
Here $\delta\neq0$, since otherwise
$\beta\neq-1/9$ would force $F_{yyy}\equiv0$.
The compatibility of relation \eqref{Eq:New_add1} with
\eqref{eq:chiellini-D-Q} and the relation $(2F_{yy}-3\sigma)Q+6\beta F_yF_{yy}D=0$,
obtained by the differentiation and elimination, gives
\begin{equation}
    (\sigma-2\delta)c_1=0,
    \qquad
    c_2=c_1v_0,
    \qquad
    c_3=c_1(v_0^2+2g_0).
\end{equation}
If $c_1=0$, these equalities imply $c_2=0$, contrary to
item~\ref{item:chiellini-two}.
If $c_1\ne0$, they give $\sigma=2\delta$ and
\begin{equation}
    D=2c_1\left(\frac23F_y+\delta y-v_0\right),
    \qquad
    Q=c_1\left(\frac23F_y+\delta y-v_0\right)^2.
\end{equation}
Consequently, $\Delta_0=c_1Q-\frac14D^2\equiv0$,
contrary to item~\ref{item:chiellini-one}.
Thus none of the alternatives of
Lemma \ref{L:Lienard_degenerate} holds.
Theorem \ref{T:Lienard_autonomous_metrisability}
therefore proves the converse statement, including the
absence of nontrivial first integrals linear in momenta.

Finally, substituting $z=(Dp_1+2Qp_2)/(2c_1p_1+Dp_2)$
into the corresponding first integrals of the Li\'enard
equation gives the stated geodesic first integrals.
They are nonconstant and homogeneous of degree zero
in the momenta, whereas $H$ has degree two.
Hence they are functionally independent of $H$
on a regular open subset.
\end{proof}

Now we turn to condition \eqref{Eq:Chiellini1}, where $\beta=0$. Recall that in this case we have $G=\kappa F_y$, where $\kappa\in\mathbb{C}\setminus\{0\}$.

\begin{theorem}\label{T:chiellini-zero-metrisability}
Let $G=\kappa F_y$, where $F,G\in\mathcal A$ are independent
of $x$ and $\kappa\in\mathbb{C}\setminus\{0\}$. Suppose that the autonomous
Li\'enard differential equation \eqref{Eq:Lienard} is metrisable
and the related metrics do not have nontrivial first integrals
linear in the momenta. Then  $F_{yyy}\not\equiv0$,
and, after a suitable choice of the additive constant in $F$,
there exist constants $\sigma\in\mathbb C\setminus\{0\}$, $c_1,c_2,c_3\in\mathbb C$
such that one of the following conditions holds.
\begin{enumerate}
    \item\label{item:chiellini-zero-one}
    $c_1\ne0$ and $F$ satisfies the second-order ODE
    \begin{equation}\label{eq:chiellini-zero-case-one}
    \begin{aligned}
        &\Bigg[
        \frac49c_1(F_y)^2
        +\frac23(\sigma c_1y-2c_2)F_y
        +2\sigma c_1F\\
        &\qquad +\sigma^2c_1y^2-4\sigma c_2y+2c_3
        +3\kappa(\sigma c_1y-2c_2)
        \Bigg]F_{yy}
        \\
        &\quad +\sigma\Bigg[
        -\frac23c_1(F_y)^2
        -(\sigma c_1y-2c_2-6\kappa c_1)F_y\\
        &\qquad -3\sigma c_1F-\frac32\sigma^2c_1y^2
        +6\sigma c_2y-3c_3
        \Bigg]=0,
    \end{aligned}
    \end{equation}
    with $V\not\equiv0$, where
    \begin{equation}\label{eq:chiellini-zero-V}
    \begin{gathered}
        V(y)=        \kappa F_y+\frac19(F_y)^2
        +\frac16\left(\sigma y-\frac{2c_2}{c_1}\right)F_y
        -\frac{\sigma}{2}F-\frac{\sigma^2}{8}y^2
        +\frac{\sigma c_2}{2c_1}y
        +\frac{c_2^2-c_1c_3}{2c_1^2}.
    \end{gathered}
    \end{equation}

    \item\label{item:chiellini-zero-two}
    $c_1=0$, $c_2\neq0$, and $F$ satisfies the first-order ODE
    \begin{equation}\label{eq:chiellini-zero-case-two}
    \begin{aligned}
        \frac{c_2}{9}(F_y)^2
        +\frac13(2\sigma c_2y-c_3+3\kappa c_2)F_y
        -\sigma c_2F
        -\frac{\sigma^2c_2}{2}y^2
        +\frac{\sigma c_3}{2}y=0.
    \end{aligned}
    \end{equation}
\end{enumerate}
Set
\begin{equation}\label{eq:chiellini-zero-D-Q}
\begin{gathered}
    D(y)=\frac43c_1F_y+\sigma c_1y-2c_2,
    \\
    Q(y)=    \frac29c_1(F_y)^2
    +\frac13(\sigma c_1y-2c_2)F_y
    +\sigma c_1F-2\kappa c_1F_y
    +\frac{\sigma^2c_1}{2}y^2-2\sigma c_2y+c_3,
\end{gathered}
\end{equation}
and
\begin{equation}\label{eq:chiellini-zero-Delta}
    \Delta_0(y)=c_1Q-\frac14D^2
    =
    \begin{cases}
        -2c_1^2V,
        & \text{in item~\ref{item:chiellini-zero-one}},\\
        -c_2^2,
        & \text{in item~\ref{item:chiellini-zero-two}}.
    \end{cases}
\end{equation}
On a neighborhood where $\Delta_0\ne0$, a related metric
is given by
\begin{equation}\label{eq:chiellini-zero-metric}
    ds^2=
    \frac{e^{-3\sigma x}}{\Delta_0^2}
    \left\{
        Q\,(dx)^2-D\,dx\,dy+c_1\,(dy)^2
    \right\}.
\end{equation}
Its geodesic flow has the Hamiltonian
\begin{equation}\label{eq:chiellini-zero-Hamiltonian}
    H=
    \frac12e^{3\sigma x}\Delta_0
    \left(
        c_1p_1^2+Dp_1p_2+Qp_2^2
    \right)
\end{equation}
and an additional first integral
\begin{equation}\label{eq:chiellini-zero-first-integral}
    J=(z+\kappa)
       \exp\left(\frac{F_y-z}{\kappa}\right),
    \qquad
    z=\frac{Dp_1+2Qp_2}{2c_1p_1+Dp_2}.
\end{equation}
This first integral is functionally independent
of $H$ on a regular open subset.

Conversely, suppose that $G=\kappa F_y$,
$\kappa\sigma\neq0$, $F_{yyy}\not\equiv0$, and all the
conditions of either item~\ref{item:chiellini-zero-one}
or item~\ref{item:chiellini-zero-two} hold.
Then the autonomous Li\'enard differential equation
\eqref{Eq:Lienard} is metrisable.
A related metric is given by
\eqref{eq:chiellini-zero-metric}. None of the related
metrics admits a nontrivial first integral linear in
momenta.
\end{theorem}

\begin{proof}
Substituting $G=\kappa F_y$ into the conditions of
Theorem~\ref{T:Lienard_autonomous_metrisability} gives
\eqref{eq:chiellini-zero-case-one} and
\eqref{eq:chiellini-zero-case-two}, together with
\eqref{eq:chiellini-zero-V}--\eqref{eq:chiellini-zero-metric}.
Lemma~\ref{L:Lienard_degenerate} excludes
$F_{yyy}\equiv0$ under the assumptions of the theorem.

Let us prove the converse. We exclude nontrivial first
integrals linear in momenta using
Lemma~\ref{L:Lienard_degenerate}.
Its first two alternatives are incompatible with
$F_{yyy}\not\equiv0$ and $G=\kappa F_y$.
Differentiating two times the relation in the third alternative  and excluding $F_y$ from these two equations, we get
\begin{equation}\label{eq:chiellini-zero-degenerate-condition}
    27\kappa\delta f_y
    =(2f+3\delta)(3\delta-f)^2,
\end{equation}
where $f=F_{yy}$ and $\delta\in\mathbb{C}$. Further, we get $\delta\ne0$. Indeed, if $\delta=0$, then
\eqref{eq:chiellini-zero-degenerate-condition} gives $2f^3=0$,
contrary to $f_y=F_{yyy}\not\equiv0$.

If $c_1\neq0$, then the compatibility of  \eqref{eq:chiellini-zero-degenerate-condition} and \eqref{eq:chiellini-zero-case-one} yields $\Delta_0=c_1Q-D^2/4\equiv0$. This fact contradicts the non-degeneracy assumption.

 If $c_1=0$, then the compatibility of  \eqref{eq:chiellini-zero-degenerate-condition} and \eqref{eq:chiellini-zero-case-two} gives
 $\sigma\delta=0$. This relation cannot hold.
Thus, none of the alternatives of
Lemma~\ref{L:Lienard_degenerate} can be realized, and
Theorem~\ref{T:Lienard_autonomous_metrisability}
proves the converse statement.

The first integral of the autonomous Li\'enard differential equation $y_{xx}=F_{yy}(y)(y_x+\kappa)$ is given in \cite{demina2025integrability}. This first integral produces first integral~\eqref{eq:chiellini-zero-first-integral} of the metric under consideration.
\end{proof}

A function is called Liouvillian over a differential field
$\mathcal K$ if it belongs to a differential field obtained from
$\mathcal K$ by a finite sequence of adjunctions of algebraic
elements, primitives, and exponentials of primitives;
see \cite{singer1992liouvillian}.


Theorems \ref{T:chiellini-metrisability} and \ref{T:chiellini-zero-metrisability} provide novel examples of integrable metrics with rational and transcendental in momenta first integrals. These first integrals are Liouvillian in momenta functions. Let us go further and consider the problem of finding metrics with non-Liouvillian in momenta first integrals.

Article \cite{demina2026non} gives several integrability criteria for autonomous Li\'{e}nard differential equations whose first integrals are non-Liouvillian functions. Constructing the intersections of these integrability conditions and the metrisability conditions of Theorem \ref{T:Lienard_autonomous_metrisability}, one can find integrable metrics possessing first integrals that are non-Liouvillian functions with respect to  momenta. Moreover, these metrics will have no linear in  momenta first integrals.  Let us present an example.

\textbf{Example 2.}
Consider the autonomous Li\'{e}nard differential equation
\begin{equation}
\label{Eq:Airy_Riemannian_Lienard}
y_{xx}=-2\psi\psi_y y_x-\psi_y,\qquad \psi=\psi(y).
\end{equation}
This family is obtained from Theorem~7
of~\cite{demina2026non} by setting $s=1$ and
relabelling the variables. Here we allow $\psi$ to be
holomorphic, using the extension described after that theorem.
Let us construct a metrisable member of this family for which
the related metric is positive definite and has no first
integral polynomial of degree at most two in momenta
and functionally independent of the Hamiltonian.

Choose representatives satisfying
$F_y=-\psi^2$, $F(0)=0$, and $G=-\psi$.
Let us use the differential condition in item~$1$ of
Theorem~\ref{T:Lienard_autonomous_metrisability} with
$\sigma=c_1=c_2=1$ and $c_3=2$.
The assumption concerning
Lemma~\ref{L:Lienard_degenerate} will be verified below.
In the notation of Theorem~\ref{T:Lienard_autonomous_metrisability}, we have
\begin{equation}
\label{Eq:Airy_Riemannian_DQ}
\begin{gathered}
D=y-2-\frac43\psi^2,\quad
Q=\frac29\psi^4-\frac13(y-2)\psi^2
+F+2\psi+\frac12y^2-2y+2,
\end{gathered}
\end{equation}
and $V=(D^2-4Q)/8$.
Differentiating the expressions for $D$ and $Q$ yields
\[
D_y=1-\frac83\psi\psi_y,\qquad
Q_y=D+\left(2-\frac23\psi D\right)\psi_y.
\]
The differential condition in item~$1$ takes the form $DV_y+(3-D_y)V=0$.
Substituting our expressions into this relation,
we obtain $\left(D+4\psi Q/3\right)\psi_y+Q=0$.
Solving the latter relation for $Q$, we get
\begin{equation}
\label{Eq:Airy_Riemannian_Q_reconstruction}
Q=-\frac{D\psi_y}{1+\frac43\psi\psi_y}.
\end{equation}
Differentiating this expression and using the
relations for $D_y$ and $Q_y$, we see that the function
$\psi$ can be defined by the Cauchy problem
\begin{equation}
\label{Eq:Airy_Riemannian_system}
\begin{gathered}
\psi_{yy}=\frac4{27}(9+8\psi^3)(\psi_y)^3
+\frac{12\psi}{6-3y+4\psi^2}(\psi_y)^2
+\left(\frac9{6-3y+4\psi^2}-2\psi\right)\psi_y-1,\\
\psi(0)=0,\qquad \psi_y(0)=1.
\end{gathered}
\end{equation}
The denominator $6-3y+4\psi^2$ equals $6$ at $y=0$,
so this problem has the unique holomorphic solution near the origin.
For this solution, the denominator in
\eqref{Eq:Airy_Riemannian_Q_reconstruction} equals $1$ at $y=0$.
Consequently, this relation defines a holomorphic function
$Q$ with $Q(0)=2$ satisfying
$Q_y=D+(2-2\psi D/3)\psi_y$.
The functions $\psi$ and $Q$ are real for real $y$, and their Taylor
expansions begin as
\begin{equation}
\label{Eq:Airy_Riemannian_series}
\begin{gathered}
\psi(y)=y+\frac{11}{12}y^2+\frac{65}{36}y^3
+\frac{3503}{864}y^4+O(y^5),\quad
Q(y)=2+3y^2+\frac{25}{6}y^3+O(y^4).
\end{gathered}
\end{equation}
Setting $\Delta_0=Q-D^2/4=-2V$, we get $\Delta_0(0)=1$. Consequently,  both conditions $\Delta_0\neq0$ and $V\neq0$
hold on a sufficiently small neighborhood of the origin.

By construction, the functions $l=s=\exp(x)$,  and $r=2\exp(x)$
satisfy $(A)_{-\frac23}$. Thus,
Theorem~\ref{T:Inv} gives the invariant
\begin{equation}
\label{Eq:Airy_Riemannian_invariant}
H=\exp(x)\left(z^2-D(y)z+Q(y)\right)
\end{equation}
with the cofactor $\lambda=2F_{yy}/3$.

Let us show that $d_{-\frac23}=1$.
For an arbitrary solution $(l,s,r)$ of $(S)_{-\frac23}$,
substituting~\eqref{Eq:Airy_Riemannian_series} into
$(A)_{-\frac23}$ and comparing the coefficients of
$1$, $y$, $y^2$, and $y^3$ yields
\begin{equation}
\label{Eq:Airy_Riemannian_mobility_conditions}
\begin{gathered}
r_x-2s=0,\qquad
9l_x-11s-6s_{xx}+4r=0,\\
8l+22l_x+3l_{xxx}-65s-12s_x+22r=0,\\
528l+975l_x-3407s-660s_x+1282r=0.
\end{gathered}
\end{equation}
Eliminating $l$ and $s$ from these relations and their
$x$-derivatives, we obtain $P_6(\partial_x)r=0$, $P_5(\partial_x)r=0$,
where
\[
\begin{aligned}
P_6(t)&=18t^6+165t^4-84t^3-343t^2+308t-64,\\
P_5(t)&=33t^5+140t^3-77t^2-75t-21.
\end{aligned}
\]
Since $\gcd(P_6,P_5)=t-1$, we get $r_x=r$.
Returning to system~\eqref{Eq:Airy_Riemannian_mobility_conditions}, we see that the latter has the unique family of solutions
$(l,s,r)=k\exp(x)(1,1,2)$, where $k\in\mathbb C$.
Consequently,
$d_{-\frac23}=1$. Invariant \eqref{Eq:Airy_Riemannian_invariant} is non-degenerate.
Thus, system $(S)_{-\frac23}$ has no nonzero degenerate solutions.
We conclude that none of the conditions of
Lemma~\ref{L:Lienard_degenerate} holds.

We can therefore apply item~$1$ of
Theorem~\ref{T:Lienard_autonomous_metrisability}.
The related metric takes the form
\begin{equation}
\label{Eq:Airy_Riemannian_metric}
ds^2=\frac{\exp(-3x)}{(\Delta_0(y))^2}
\left\{Q(y)(dx)^2-D(y)\,dx\,dy+(dy)^2\right\},
\end{equation}
where
\[
\Delta_0(y)=
\frac{\left(2-y+\frac43\psi^2\right)\psi_y}
{1+\frac43\psi\psi_y}
-\frac14\left(2-y+\frac43\psi^2\right)^2.
\]
Note that $\Delta_0(0)=1$. It is straightforward to find
\[
g_{22}=\frac{\exp(-3x)}{\Delta_0^2}\neq0,\qquad
\det(g)=\frac{\exp(-6x)}{\Delta_0^3}.
\]
For real $x$ and $y$,  metric \eqref{Eq:Airy_Riemannian_metric} is positive definite
on a sufficiently small neighborhood of the origin, since
\[
g(0,0)=
\begin{pmatrix}
2&1\\
1&1
\end{pmatrix}.
\]
The related Hamiltonian reads as
\begin{equation}
\label{Eq:Airy_Riemannian_Hamiltonian}
H=
\frac{\exp(3x)\Delta_0(y)}2
\left(p_x^2+D(y)p_xp_y+Q(y)p_y^2\right).
\end{equation}

Let us find a  first integral of the
related geodesic flow. It follows from Theorem~7 of~\cite{demina2026non} that
equation \eqref{Eq:Airy_Riemannian_Lienard} has the first integral
\begin{equation}
\label{Eq:Airy_Riemannian_ODE_integral}
I(y,z)=
\frac{\psi(y)\operatorname{Ai}(\xi)+\operatorname{Ai}'(\xi)}
{\psi(y)\operatorname{Bi}(\xi)+\operatorname{Bi}'(\xi)},
\qquad \xi=z+\psi(y)^2,
\end{equation}
where primes denote differentiation with respect to the
arguments of the Airy functions.
Using~\eqref{Eq:Metric_FI_Tr}, we obtain the first integral
\begin{equation}
\label{Eq:Airy_Riemannian_geodesic_integral}
\begin{gathered}
J=
\frac{\psi(y)\operatorname{Ai}(\eta)+\operatorname{Ai}'(\eta)}
{\psi(y)\operatorname{Bi}(\eta)+\operatorname{Bi}'(\eta)},\quad
\eta=\psi(y)^2+
\frac{D(y)p_x+2Q(y)p_y}{2p_x+D(y)p_y}
\end{gathered}
\end{equation}
of the geodesic flow of~\eqref{Eq:Airy_Riemannian_metric}.
The function $J$ is non-constant and homogeneous
of degree zero in momenta. Hence, it is functionally
independent of the Hamiltonian $ H$.
The Gaussian curvature of metric \eqref{Eq:Airy_Riemannian_metric} reads as
\begin{equation}
\label{Eq:Airy_Riemannian_curvature}
K(x,y)=e^{3x}\Delta_0(y)
\left[
\left(1+\frac{\psi D}{3}\right)\psi_{yy}
+\left(\frac{D}{3}-\frac49\psi^2-\frac1Q\right)\psi_y^2
+\frac13\psi\psi_y+\frac12
\right].
\end{equation}

Other examples of integrable metrics admitting first integrals transcendental in  momenta can be found in \cite{galajinsky2021some, agapov2023new, demina2025polynomial}.

\section{Li\'{e}nard differential equations with \texorpdfstring{$f(x,y)\equiv0$}{f(x,y) = 0}}\label{S:f_zero}

Let $f(x,y)\equiv0$. Then it is without loss of generality to set $F(x,y)\equiv0$. The vector field  \eqref{Eq:Lienard_vector_field} is divergence-free. Any invariant of Theorem \ref{T:Inv} is a first integral of the related equation. Obviously, condition $(A)_{\alpha}$ and system $(S)_{\alpha}$ are independent of $\alpha$ whenever $F(x,y)\equiv0$. Thus, there is no loss of generality in omitting the index  $\alpha$. Further in this section, we will refer to them simply as  condition $(A)$ and system $(S)$.

If a Li\'{e}nard differential equation~\eqref{Eq:Lienard} possesses a non-degenerate invariant quadratic in $z$,  then it is metrisable provided that $F(x,y)\equiv0$. The related metrics always have a linear in  momenta first integral, since the equation under consideration has a degenerate invariant $H=1$, see Theorem \ref{T:Metrics_General} for details.

Let us study the existence of non-constant first integrals. Considering condition $(A)$ as a partial differential equation with respect to $G(x,y)$, we integrate it. As a result, we get
  \begin{equation}
\label{Eq:G_exact}
\begin{gathered}
G(x,y)=\left(\frac{l_{xx}}{4l}-\frac{(l_x)^2}{8l^2}\right)y^2-\sqrt{l}\left(w_{xx}
+\frac{l_xw_x}{l}\right)y+\frac{r}{2l}-\frac{l(w_x)^2}{2}+\frac{1}{l}L\left(w+\frac{y}{\sqrt{l}}\right),\, l\not\equiv0;\\
G(x,y)=\frac{s_{xx}}{2s}y^2-\frac{r_x}{2s}y+L(x),\quad l\equiv0, \quad s\not\equiv0.
\end{gathered}
\end{equation}
Here $L$ is an arbitrary function of one variable and the function $w(x)$ is defined via the relation $w_x=sl^{-\frac32}$. These expressions are considered locally where $l\neq0$ or $s\neq0$, respectively, with a fixed branch of $\sqrt l$ in the first case. The function $L$ is chosen so that the displayed expressions are holomorphic on the domain under consideration. The remaining case $l\equiv s\equiv0$ gives $r\in\mathbb{C}$ and the constant first integral $H=r$, with no restriction on $G$.

As an example, let us solve the metrisability problem for the Ermakov--Pinney equation \cite{ermakov1880transformation, pinney1950nonlinear, milne1930numerical}  given by
\begin{equation}
\label{Eq:Ermakov--Pinney}
y_{xx}= u(x) y + \frac{v_0}{ y^3}.
\end{equation}
Here $u\in\mathcal B$ and $v_0\in\mathbb{C}\setminus\{0\}$.  It is well-known that the  Ermakov--Pinney equation has two functionally independent first integrals that are quadratic with respect to $z=y_x$, see, for example \cite{polyanin2003handbook}.  These first integrals are sometimes called Ermakov--Lewis invariants.
They are explicitly given in  \eqref{Eq:Ermakov--Pinney_Cond1_FI2}.

\begin{theorem}\label{T:Ermakov--Pinney_metrizability}
Let $u\in\mathcal B$ and $v_0\in\mathbb{C}\setminus\{0\}$. The  Ermakov--Pinney equation \eqref{Eq:Ermakov--Pinney} is locally metrisable. The related family of metrics takes the form
\begin{equation}
\label{Eq:Metric_metrisability_Ermakov--Pinney}
ds^2=\frac{1}{\Delta^2}\left[\left(\frac{l_{xx}}{2}-ul\right)y^2+\frac{lv_0}{y^2}
+c_4\right] (dx)^2-\frac{l_xy}{\Delta^2}dxdy+\frac{l}{\Delta^2}(dy)^2,
\end{equation}
where
\begin{equation}
\label{Eq:Metric_metrisability_Ermakov--Pinney_del}
\Delta=\left(\frac{ll_{xx}}{2}-\frac{(l_x)^2}{4}-ul^2\right)y^2+c_4l+\frac{l^2v_0}{y^2}.
\end{equation}
Here $l$ is given by
\begin{equation}
\label{Eq:Ermakov--Pinney_Case2_l}
\begin{gathered}
l(x)=c_1(w_1)^2+c_2(w_2)^2+c_3w_1w_2,
\end{gathered}
\end{equation}
 the functions $w_1$ and $w_2$ are two linearly independent solutions of the linear Schr\"{o}dinger equation $w_{xx}-uw=0$, $c_j\in\mathbb{C}$, $1\leq j \leq 4$, and  $|c_1|+|c_2|+|c_3|>0$. The metrics are considered on neighborhoods, where
$y\neq0$ and $\Delta\neq0$.
\end{theorem}

\begin{proof}
Since $F=0$, it follows from Theorem~\ref{T:Metrisability}
that the Ermakov--Pinney equation~\eqref{Eq:Ermakov--Pinney}
is locally metrisable if and only if it has a first integral
of degree at most two in $z$ such that
$\Delta=C_0C_2-C_1^2$ does not vanish at the point under
consideration. Such a first integral is necessarily nonconstant.
Let us study the existence of these first integrals.

A representative of $G\in\mathcal A$ reads as
\begin{equation}
\label{Eq:Ermakov--Pinney_System_G}
G=\frac{uy^2}{2}-\frac{v_0}{2y^2}.
\end{equation}
The explicit expressions for the coefficients of quadratic
first integrals~\eqref{Eq:genLS} were found in
Theorem~\ref{T:Inv}. In the case under consideration they
take the form
\begin{equation}
\label{Eq:Ermakov--Pinney_Coeff_FI}
\begin{gathered}
C_0=l,\qquad
C_1=s-\frac{l_x}{2}y,\qquad
C_2\equiv\Psi=\frac{l_{xx}}{2}y^2-2s_xy+r-2lG.
\end{gathered}
\end{equation}
Consequently, first integrals of degree at most two in $z$
can be represented as
\begin{equation}
\label{Eq:Ermakov--Pinney_FI_explicit}
H=lz^2+2\left(s-\frac{l_x}{2}y\right)z
+\frac{l_{xx}}{2}y^2-2s_xy+r-2lG.
\end{equation}

Substituting~\eqref{Eq:Ermakov--Pinney_System_G} into
condition~$(A)$, we get the equation
\begin{equation}
\label{Eq:Ermakov--Pinney_System1}
\frac12\left(l_{xxx}-4ul_x-2u_xl\right)y^2
-2(s_{xx}-us)y+r_x+\frac{2v_0s}{y^3}=0.
\end{equation}
Comparing the coefficients of distinct powers of $y$ and
using $v_0\ne0$, we conclude that system~$(S)$ is equivalent to
\begin{equation}
\label{Eq:Ermakov--Pinney_System1_Coeff}
l_{xxx}-4ul_x-2u_xl=0,\qquad s=0,\qquad r_x=0.
\end{equation}
The equation $s_{xx}-us=0$ is automatically satisfied.

The first equation
in~\eqref{Eq:Ermakov--Pinney_System1_Coeff}
is a linear homogeneous third-order ordinary differential
equation sometimes called Appell's equation. If $w_1$ and
$w_2$ are two linearly independent solutions of
$w_{xx}-uw=0$, then $(w_1)^2$, $(w_2)^2$, and $w_1w_2$
form a fundamental system of solutions of Appell's equation.
Thus, its general solution is given
by~\eqref{Eq:Ermakov--Pinney_Case2_l}.
The remaining equations yield $s=0$ and $r=c_4$,
where $c_4\in\mathbb C$. Consequently, the index $d$
of system~$(S)$ equals $4$.

If $l\equiv0$, then $C_0=C_1=0$ and $\Delta\equiv0$.
Hence, only solutions with $l\not\equiv0$ can give rise
to metrics. This condition is equivalent to
$|c_1|+|c_2|+|c_3|>0$.

Setting $c_1=1$, $c_2=c_3=c_4=0$, and then
$c_2=1$, $c_1=c_3=c_4=0$, we obtain the first integrals
\begin{equation}
\label{Eq:Ermakov--Pinney_Cond1_FI2}
I_j=\left(w_jz-(\partial_xw_j)y\right)^2
+\frac{v_0(w_j)^2}{y^2},
\qquad j=1,2.
\end{equation}
Their functional independence on a regular open subset
follows from the non-vanishing Wronskian of $w_1$ and $w_2$.

For the solutions of system~$(S)$ found above, we have
\[
\Psi=\left(\frac{l_{xx}}2-ul\right)y^2
+\frac{lv_0}{y^2}+c_4,
\qquad
\Delta=l\Psi-\frac{(l_x)^2y^2}{4}.
\]
Expanding the latter expression
gives~\eqref{Eq:Metric_metrisability_Ermakov--Pinney_del}.

Let us verify that a non-degenerate solution exists near
every point $(x_0,y_0)$ where the equation is defined.
Since $w_1$ and $w_2$ have a nonzero Wronskian, at least
one of them satisfies $w_j(x_0)\ne0$.
Choosing $l=(w_j)^2$ and $c_4=0$, we obtain $\Delta=v_0(w_j)^4/y^2$.
Since $v_0\ne0$ and $y_0\ne0$, this expression does not
vanish at $(x_0,y_0)$. Theorem~\ref{T:Metrisability}
therefore implies local metrisability.

Finally, substituting
$C_0=l$, $C_1=-l_xy/2$, and $C_2=\Psi$
into \eqref{Eq:Metric_metrisability},
we find
\[
ds^2=\frac{1}{\Delta^2}
\left\{\Psi\,(dx)^2-l_xy\,dx\,dy+l\,(dy)^2\right\}.
\]
This is precisely
\eqref{Eq:Metric_metrisability_Ermakov--Pinney}.
Since all solutions of system~$(S)$ have been found,
the displayed formula gives the entire local family
of related metrics on neighborhoods where $\Delta\ne0$.
\end{proof}

Finally, we note that all metrics of Theorem \ref{T:Ermakov--Pinney_metrizability} have a linear first integral. By Theorem~\ref{T:Metrics_General} this first integral takes the form $J=lp_1+l_xyp_2/2$,
where the explicit expression of $l$ is given in  Theorem \ref{T:Ermakov--Pinney_metrizability}.

\section{Conclusion}\label{S:Conclusion}

In this article we have obtained the necessary and sufficient conditions for
Li\'{e}nard differential equations \eqref{Eq:Lienard_Int} to have quadratic
 invariants with the cofactor proportional to the coefficient of $y_x$, i.e. $f(x,y)$.
These conditions in the case of the cofactor $2f(x,y)/3$, together with
non-degeneracy, give a criterion for the metrisability of the equation under study.
In the autonomous case, they provide a constructive description
of the related metrics. We have considered in detail the case of metrics without first integrals that are linear in momenta.
We have also described the relationship
between first integrals of metrisable equations and those of
the corresponding geodesic flows.

We note that our method of solving the metrisability problem
for Li\'enard differential equations \eqref{Eq:Lienard_Int}
differs from the general prolongation method of
Bryant et al.~\cite{bryant2009metrisability}.
Both approaches are based on Liouville's metrisability system  \eqref{Eq:Geodesic_main_PDE}.
For prescribed $F$ and $G$, our method reduces the determining
equations to a system of linear ordinary differential equations
for three functions of one variable, possibly supplemented by
algebraic constraints. This reduction facilitates explicit
computations and provides a constructive procedure for
finding the related metrics.

Using our solution of the metrisability problem for autonomous Li\'{e}nard differential equations, we have found integrable metrics with rational and transcendental in momenta first integrals.  These metrics are generated by Li\'{e}nard  equations satisfying the Chiellini integrability condition \cite{Chiellini01}.  In addition, combining our metrisability results and some integrability conditions from \cite{demina2026non},  we have constructed a positive-definite metric whose geodesic flow possesses a first integral that is non-Liouvillian in the momentum ratio. This metric has no first integrals polynomial of degree at most two in  momenta and functionally independent of the Hamiltonian.  Articles~\cite{demina2025integrability, demina2026non} contain some criteria for the integrability of autonomous Li\'{e}nard differential equations. Combined
with our metrisability conditions, these criteria can be used to construct
other metrics with transcendental in
momenta first integrals.

Any analytic second-order ODE of the form \eqref{Eq:ODE_projective} can be locally reduced near a regular point to a non-autonomous Li\'{e}nard differential equation \eqref{Eq:Lienard_Int} by a non-singular point transformation \cite{bryant2009metrisability}. Since metrisability is invariant under point transformations, our method applies to the general local analytic metrisability problem after this reduction. However, constructing the required transformation may itself involve solving nontrivial auxiliary differential equations. Thus, an interesting open problem is to explicitly solve the metrisability problem for other relevant families of second-order ordinary differential equations of the form~\eqref{Eq:ODE_projective}.

\section*{Acknowledgments}
We are grateful to S. V. Agapov for useful discussions.
This work is supported by Russian Science Foundation grant
25-71-10087, \url{https://rscf.ru/en/project/25-71-10087/}.

\section*{Data availability}
This is a theoretical study. The mathematical results and derivations
supporting its conclusions are contained in the article.

\bibliography{references}

@article{liouville1889invariants,
  title={Sur les invariants de certaines {\'e}quations diff{\'e}rentielles et sur leurs applications},
  author={Liouville, Roger},
  journal={Journal de l'{\'E}cole polytechnique},
  volume={59},
  pages={7--76},
  year={1889}
}

@article{bryant2009metrisability,
  title={Metrisability of two-dimensional projective structures},
  author={Bryant, Robert and Dunajski, Maciej and Eastwood, Michael},
  journal={Journal of differential geometry},
  volume={83},
  number={3},
  pages={465--500},
  year={2009},
  publisher={Lehigh University}
}

@article{agapov2024integrable,
  title={Integrable geodesic flows and metrisable second-order ordinary differential equations},
  author={Agapov, Sergei V and Demina, Maria V},
  journal={Journal of Geometry and Physics},
  volume={199},
  pages={105168},
  year={2024},
  publisher={Elsevier}
}

@article{contatto2018metrisability,
  title={Metrisability of {P}ainlev{\'e} equations},
  author={Contatto, Felipe and Dunajski, Maciej},
  journal={Journal of Mathematical Physics},
  volume={59},
  number={2},
  pages={023507},
  year={2018},
  publisher={AIP Publishing}
}

@article{nucci2008jacobi,
  title={The {J}acobi Last Multiplier and its applications in mechanics},
  author={Nucci, Maria Clara and Leach, PGL},
  journal={Physica Scripta},
  volume={78},
  number={6},
  pages={065011},
  year={2008}
}

@article{Koenigs01,
  title={Sur les g\'{e}odesiques a int\'{e}grales quadratiques},
  author={Koenigs, M.G.},
  journal={Note II from Darboux': Le\c{c}ons sur la th\'{e}orie g\'{e}n\'{e}rale des surfaces},
  volume={IV},
   year={1896},
  publisher={Chelsea Publishing}
}

@article{matveev2011two,
  title={Two-dimensional superintegrable metrics with one linear and one cubic integral},
  author={Matveev, Vladimir S and Shevchishin, Vsevolod V},
  journal={Journal of Geometry and Physics},
  volume={61},
  number={8},
  pages={1353--1377},
  year={2011},
  publisher={Elsevier}
}

@article{demina2023integrability,
  title={Integrability and solvability of polynomial {L}i{\'e}nard differential systems},
  author={Demina, Maria V},
  journal={Studies in Applied Mathematics},
  volume={150},
  number={3},
  pages={755--817},
  year={2023},
  publisher={Wiley Online Library}
}

@article{Llibre06,
  title   = "Liouvillian first integrals for generalized {L}i\'{e}nard polynomial differential systems",
  journal = "Adv. Nonlinear Stud.",
  volume  = "13",
   pages   = "819--829",
  year    = "2013",
 author  = "Llibre, J. and Valls, C."
}

@article{demina2025integrability,
  title={Integrability Properties of Generalized {L}i{\'e}nard Differential Equations},
  author={Demina, Maria V and Nechitailo, Varvara G},
  journal={Qualitative Theory of Dynamical Systems},
  volume={24},
  number={1},
  pages={1--27},
  year={2025},
  publisher={Springer}
}

@article{ruiz2018integrability,
  title={On the integrability of {L}i{\'e}nard {I}-type equations via $\lambda$-symmetries and solvable structures},
  author={Ruiz, Adrian and Muriel, C},
  journal={Applied Mathematics and Computation},
  volume={339},
  pages={888--898},
  year={2018},
  publisher={Elsevier}
}

@article{la2026c,
  title={${C}^{\infty}$-Structures for {L}i{\'e}nard Equations and New Exact Solutions to a Class of {K}lein--{G}ordon Equations},
  author={la Flor, Beltr{\'a}n de and Ruiz, Adri{\'a}n and Muriel, Concepci{\'o}n},
  journal={Mathematical Methods in the Applied Sciences},
  volume={49},
  number={4},
  pages={2795--2822},
  year={2026},
  publisher={Wiley Online Library}
}

@article{sinelshchikov2019lax,
  title={Lax representation and quadratic first integrals for a family of non-autonomous second-order differential equations},
  author={Sinelshchikov, Dmitry I and Gaiur, Ilia Yu and Kudryashov, Nikolay A},
  journal={Journal of Mathematical Analysis and Applications},
  volume={480},
  number={1},
  pages={123375},
  year={2019},
  publisher={Elsevier}
}

@article{carinena2021jacobi,
  title={Jacobi multipliers in integrability and the inverse problem of mechanics},
  author={Cari{\~n}ena, Jos{\'e} F and Fern{\'a}ndez-N{\'u}{\~n}ez, Jos{\'e}},
  journal={Symmetry},
  volume={13},
  number={8},
  pages={1413},
  year={2021},
  publisher={MDPI}
}

@article{matveev2012two,
  title={Two-dimensional metrics admitting precisely one projective vector field},
  author={Matveev, Vladimir S},
  journal={Mathematische Annalen},
  volume={352},
  number={4},
  pages={865--909},
  year={2012},
  publisher={Springer}
}

@article{agapov2024finite,
  title={Finite-gap potentials and integrable geodesic equations on a 2-surface},
  author={Agapov, Sergei Vadimovich and Mironov, Andrei Evgen'evich},
  journal={Proceedings of the Steklov Institute of Mathematics},
  volume={327},
  number={1},
  pages={1--11},
  year={2024},
  publisher={Springer}

}

@article{ermakov1880transformation,
  title={Transformation of differential equations},
  author={Ermakov, V. P.},
  journal={Univ. Izv. Kiev},
  volume={20},
  number={1},
  year={1880}
}

@book{polyanin2003handbook,
  title={Handbook of exact solutions for ordinary differential equations},
  author={Polyanin, Andrei D and Zaitsev, Valentin F},
  volume={83},
  year={2003},
  publisher={Chapman \& Hall/CRC Boca Raton}
}

@article{milne1930numerical,
  title={The numerical determination of characteristic numbers},
  author={Milne, WE},
  journal={Physical Review},
  volume={35},
  number={7},
  pages={863},
  year={1930},
  publisher={APS}
}

@article{pinney1950nonlinear,
  title={The nonlinear differential equation {$y''+p(x)y+cy^{-3}=0$}},
  author={Pinney, Edmund},
  journal={Proceedings of the American Mathematical Society},
  volume={1},
  number={5},
  pages={681},
  year={1950},
  doi={10.1090/S0002-9939-1950-0037979-4}
}

@article{demina2026non,
  title={From non-polynomial invariants to non-Liouvillian first integrals},
  author={Demina, Maria and Nechitailo, Varvara},
  journal={Analysis and Mathematical Physics},
  volume={16},
  number={5},
  pages={113},
  year={2026},
  publisher={Springer}
}

@article{muriel2011second,
  title={Second-order ordinary differential equations with first integrals of the form $C (t)+ 1/(A (t, x) \dot{x}+ B (t, x))$},
  author={Muriel, C and Romero, JL},
  journal={Journal of Nonlinear Mathematical Physics},
  volume={18},
  number={supp01},
  pages={237--250},
  year={2011},
  publisher={World Scientific}
}

@article{muriel2009second,
  title={Second-order ordinary differential equations and first integrals of the form $A(t, x) \dot{x}+ B (t, x)$},
  author={Muriel, C and Romero, JL},
  journal={Journal of Nonlinear Mathematical Physics},
  volume={16},
  number={Suppl 1},
  pages={209--222},
  year={2009},
  publisher={Springer}
}

@article{kruglikov2008invariant,
  title={Invariant characterization of {L}iouville metrics and polynomial integrals},
  author={Kruglikov, Boris},
  journal={Journal of Geometry and Physics},
  volume={58},
  number={8},
  pages={979--995},
  year={2008},
  publisher={Elsevier}
}

@article{agapov2023new,
  title={New examples of non-polynomial integrals of two-dimensional geodesic flows},
  author={Agapov, Sergei and Shubin, Vladislav},
  journal={Journal of Physics A: Mathematical and Theoretical},
  volume={57},
  number={1},
  pages={015204},
  year={2023},
  publisher={IOP Publishing}
}

@article{galajinsky2021some,
  title={Some metrics admitting nonpolynomial first integrals of the geodesic equation},
  author={Galajinsky, Anton},
  journal={Physics Letters B},
  volume={820},
  pages={136483},
  year={2021},
  publisher={Elsevier}
}

@article{demina2025polynomial,
  title={Polynomial and non-polynomial first integrals of projective structures and geodesic flows},
  author={Demina, Maria V and Ishchenko, Anna R},
  journal={Journal of Geometry and Physics},
  pages={105596},
  year={2025},
  publisher={Elsevier}
}

@article{Chiellini01,
  title   = {Sull\'{i}ntegrazione della equazione differenziale
$y'+{P}y^2+{Q}y^3=0$},
  journal = {Bollettino
della Unione Matematica Italiana},
  volume  = {10},
      pages   = {301--307},
  year    = {1931},
 author  = {Chiellini,  A.}
}

@article{harko2014class,
  title={A class of exact solutions of the Li{\'e}nard-type ordinary nonlinear differential equation},
  author={Harko, Tiberiu and Lobo, Francisco SN and Mak, MK},
  journal={Journal of Engineering Mathematics},
  volume={89},
  number={1},
  pages={193--205},
  year={2014},
  publisher={Springer}
}

@article{singer1992liouvillian,
  author  = {Singer, Michael F.},
  title   = {Liouvillian First Integrals of Differential Equations},
  journal = {Transactions of the American Mathematical Society},
  volume  = {333},
  number  = {2},
  pages   = {673--688},
  year    = {1992}
}

\end{document}